\documentclass[acmlarge]{acmart}
\usepackage{fullpage}
\usepackage{xcolor}
\usepackage{times}
\makeatletter
\def\@copyrightspace{\relax}
\makeatother
\usepackage{booktabs} 
\usepackage{mdframed}
\usepackage{framed}
\usepackage{hyphenat}
\usepackage{algpseudocode}
\usepackage{amsmath,epsfig}
\usepackage[boxed,ruled,vlined,linesnumbered]{algorithm2e}
\usepackage{amsthm}
\usepackage{setspace}
\usepackage{enumitem}
\usepackage{ragged2e}
\usepackage{multirow}
\usepackage{hhline}
\usepackage{wrapfig}
\usepackage{amssymb}
\usepackage{caption}
\usepackage{subcaption}
\usepackage{graphicx}
\usepackage{amsmath,epsfig}
\usepackage{amsthm}

\setcopyright{none}
\setcopyright{acmcopyright}
\setcopyright{acmlicensed}
\setcopyright{rightsretained}
\setcopyright{usgov}
\setcopyright{usgovmixed}
\setcopyright{cagov}
\setcopyright{cagovmixed}

\newcommand{\FFF}{\vspace*{\smallskipamount}}

\acmJournal{TMIS}
\acmVolume{Unassigned}
\acmNumber{Unassigned}
\acmMonth{5}

\begin{document}

\title{\textsc{Panda}: Partitioned Data Security on Outsourced Sensitive and Non-sensitive Data}\titlenote{\textbf{A preliminary version of this paper was accepted and presented in IEEE ICDE 2019~\cite{DBLP:conf/icde/Mehrotra0UM19}.} \newline
\textbf{This version has been accepted in ACM Transactions on Management Information Systems.} The final published version of this paper may differ
from this accepted version. {\newline This material is based on research sponsored by DARPA under agreement number FA8750-16-2-0021. The U.S. Government is authorized to reproduce and distribute reprints for Governmental purposes notwithstanding any copyright notation thereon. The views and conclusions contained herein are those of the authors and should not be interpreted as necessarily representing the official policies or endorsements, either expressed or implied, of DARPA or the U.S. Government. This work is partially supported by NSF grants 1527536 and 1545071.}}

\author{Sharad Mehrotra$^1$, Shantanu Sharma$^1$, Jeffrey D. Ullman$^2$, \\ Dhrubajyoti Ghosh$^1$, Peeyush Gupta$^1$}

\affiliation{\\
$^1$University of California, Irvine, USA. $^2$Stanford University, USA.}

\begin{abstract}
Despite extensive research on cryptography, secure and efficient query processing over outsourced data remains an open challenge. This paper continues along with the emerging trend in secure data processing that recognizes that the entire dataset may not be sensitive, and hence, non-sensitivity of data can be exploited to overcome limitations of existing encryption-based approaches. We, first, provide a new security definition, entitled \emph{partitioned data security} for guaranteeing  that the joint processing of non-sensitive data (in cleartext) and sensitive data (in encrypted form) does not lead to any leakage. Then, this paper proposes a new secure approach, entitled \emph{query binning} (QB) that allows secure execution of queries over non-sensitive and sensitive parts of the data. QB maps a query to a set of queries over the sensitive and non-sensitive data in a way that no leakage will occur due to the joint processing over sensitive and non-sensitive data. {In particular, we propose secure algorithms for selection, range, and join queries to be executed over encrypted sensitive and cleartext non-sensitive datasets. }Interestingly, in addition to improving performance, we show that QB actually strengthens the security of the underlying cryptographic technique by preventing size, frequency-count, and workload-skew attacks.
 \end{abstract}

%

\maketitle

{\color{black}
\section{Introduction}
\label{sec:introduction}

{
The past two decades have witnessed the emergence of several public clouds (\textit{e}.\textit{g}., Amazon Web Services, Google APP Engine, and Microsoft Azure) as the dominant computation, storage, and data management platform. Many small to medium size organizations, including some large organizations such as Netflix, have adopted the cloud model shifting their data management task to the cloud. The cloud offers numerous advantages including the economy of scale, low barriers to entry, limitless scalability, and a pay-as-you-go model.
While benefits abound, a key challenge from data owners' perspective -- that of ``losing'' control over one's data -- still plagues the cloud model. In addition, the threat of ``insider attacks'' is also realistic, loss of control can lead to significant security, privacy, and confidentiality concerns. Such concerns are not a new revelation --- indeed, they were identified as a key impediment for organizations adopting the {\em database-as-as-service} model in early work on data outsourcing \cite{DBLP:conf/sigmod/HacigumusILM02}. Since then, the security/confidentiality challenge has been extensively studied in both the cryptography and database literature. Existing work on data security can be broadly categorized into the following three categories:
\begin{enumerate}[noitemsep,nolistsep,leftmargin=.01cm]
  \item \textbf{Encryption based techniques.} These techniques include order-preserving encryption (OPE)~\cite{DBLP:conf/sigmod/AgrawalKSX04}, deterministic encryption~\cite{DBLP:conf/crypto/BellareBO07}, non-deterministic encryption~\cite{DBLP:journals/jcss/GoldwasserM84}, homomorphic encryption~\cite{gentry2009fully}, bucketization~\cite{DBLP:conf/sigmod/HacigumusILM02}, searchable encryption~\cite{DBLP:conf/sp/SongWP00,DBLP:journals/jcs/CurtmolaGKO11},  and distributed searchable symmetric encryption (DSSE)~\cite{DBLP:conf/ctrsa/IshaiKLO16}). In addition, encryption-based techniques resulted in the following secure database systems:
CryptDB~\cite{DBLP:journals/cacm/PopaRZB12}, Monomi~\cite{popa-monomi}, TrustedDB~\cite{DBLP:journals/tkde/BajajS14}, CorrectDB~\cite{DBLP:journals/pvldb/BajajS13}, SDB~\cite{DBLP:conf/sigmod/WongKCLY14}, ZeroDB~\cite{DBLP:journals/corr/EgorovW16}, L-EncDB~\cite{DBLP:journals/kbs/0002LCXTW15}, MrCrypt~\cite{DBLP:conf/oopsla/TetaliLMM13}, Crypsis~\cite{crypsis}, Arx~\cite{arx-popa-2017}. Likewise,  Cypherbase~\cite{DBLP:conf/cidr/ArasuBEKKRV13}, Microsoft Always Encrypted, Oracle 12c, Amazon Aurora~\cite{aurora}, and MariaDB~\cite{mariadb} are industrial secure encrypted databases.

  \item \textbf{Secret-sharing (SS)~\cite{DBLP:journals/cacm/Shamir79} based techniques.} Examples of which include distributed point function~\cite{DBLP:conf/eurocrypt/GilboaI14}, function secret sharing~\cite{DBLP:conf/eurocrypt/BoyleGI15},  accumulating-automata~\cite{DBLP:conf/ccs/DolevGL15,DBLP:conf/dbsec/DolevL016}, secure secret-shared MapReduce~\cite{DBLP:conf/dbsec/DolevL016}, \textsc{Obscure}~\cite{DBLP:journals/pvldb/GuptaLMP0A19}, and  others~\cite{DBLP:journals/isci/EmekciMAA14,DBLP:conf/fc/LueksG15,DBLP:journals/iacr/LiMD14}.
In this category, two emerging industrial systems are: {\em Pulsar}\footnote{\url{https://www.stealthsoftwareinc.com/}} based on function secret-sharing and {\em Jana}~\cite{jana} based on non-deterministic, order-preserving encryption and secret-sharing.

  \item \textbf{Trusted hardware-based techniques.} They are either based on a secure coprocessor or Intel Software Guard Extensions (SGX)~\cite{sgx} that allow decrypting data in a secure area at the cloud and perform the computation on decrypted data. However, the secure coprocessor reveals access-patterns. Cipherbase~\cite{DBLP:conf/fpl/ArasuEKKRV13,DBLP:conf/cidr/ArasuBEKKRV13}
      CorrectDB~\cite{DBLP:journals/pvldb/BajajS13},
      VC3~\cite{DBLP:conf/sp/SchusterCFGPMR15}, Opaque~\cite{DBLP:journals/iacr/WangDDB06}, HardIDX~\cite{DBLP:conf/dbsec/FuhryBB0KS17},
      EnclaveDB~\cite{DBLP:conf/sp/PriebeVC18},
      Oblix~\cite{DBLP:conf/sp/MishraPCCP18},
      Hermetic~\cite{xuhermetic}, and
      EncDBDB~\cite{DBLP:journals/corr/abs-2002-05097} are secure hardware-based systems.
\end{enumerate}

\medskip
Despite significant progress, a cryptographic approach that is both \emph{secure} (\textit{i}.\textit{e}., no leakage of sensitive data to the adversary) and \emph{efficient} (in terms of time) simultaneously has proved to be very challenging. Existing solutions suffer from the following limitations:

\begin{itemize}[noitemsep,nolistsep,leftmargin=.01cm]
  \item \textbf{Non-scalability.} Cryptographic approaches that prevent leakage, \textit{e}.\textit{g}., fully homomorphic encryption coupled with oblivious-RAM (ORAM)~\cite{DBLP:conf/stoc/Goldreich87,DBLP:conf/stoc/Ostrovsky90} or secret-sharing, simply do not scale to large data sets and complex queries. Most of the above-mentioned techniques do not work well, when deployed on a large-scale dataset, due to the high overheads of the techniques. For example,
      executing a selection query on 1M TPC-H LineItem table took (\textit{i}) 22 seconds using secret-sharing based techniques~\cite{DBLP:conf/dbsec/DolevL016},
      (\textit{ii}) 797 seconds using multiparty computation-based industrial database system, namely Jana~\cite{jana}, and
      (\textit{iii}) 13 seconds using Intel SGX-based Opaque~\cite{opaque}, while executing the same query on cleartext data took only 0.0002 seconds.

  \item \textbf{Unclear security properties.} Systems such as CryptDB have tried to take a more practical approach by allowing users to explore the tradeoffs between the system functionality and the security it offers.
      Unfortunately, precisely characterizing the security offered by such systems given the underlying cryptographic approaches has turned out to be extremely difficult. For instance,~\cite{DBLP:conf/ccs/NaveedKW15,DBLP:conf/ccs/KellarisKNO16} show that when order-preserving and deterministic encryption techniques are used together, on a dataset (in which the entropy of the values is not high enough), an attacker might be able to construct the entire plaintext by doing a frequency analysis of the encrypted data.

\item \textbf{Vulnerability to other security attacks.} Many of the above-mentioned cryptographic techniques/systems are also
     susceptible to the following attacks:

\begin{enumerate}[noitemsep,nolistsep,leftmargin=0.2in]
\item\emph{Output-size attack}: An adversary having some background knowledge can deduce the full/partial outputs by simply observing the output sizes~\cite{opaque}. All the above techniques/systems, except bucketization~\cite{DBLP:conf/icde/HacigumusMI02}, are prone to output-size attacks. 

\item\emph{Frequency attack}: {An adversary can deduce how many tuples have an identical value~\cite{DBLP:conf/ccs/NaveedKW15} based on the output of the query. Order-preserving encryption~\cite{DBLP:conf/sigmod/AgrawalKSX04}, searchable encryption~\cite{DBLP:conf/sp/SongWP00}, and secure hardware-based techniques~\cite{DBLP:conf/fpl/ArasuEKKRV13,DBLP:conf/cidr/ArasuBEKKRV13,DBLP:journals/pvldb/BajajS13,DBLP:conf/sp/SchusterCFGPMR15,DBLP:conf/uss/DinhSCOZ15,To-2016,opaque,DBLP:conf/sp/PriebeVC18} are prone to frequency-count attacks, during query execution. In addition, deterministic encryption~\cite{DBLP:conf/crypto/BellareBO07} reveals the frequency of a word, even without executing a query.}

\item\emph{Access-pattern attack}: An adversary can know the addresses of encrypted tuples that satisfy the query~\cite{DBLP:journals/jacm/ChorKGS98}. Private information retrieval (PIR)~\cite{DBLP:journals/jacm/ChorKGS98},
oblivious-RAM (ORAM)~\cite{DBLP:conf/stoc/Goldreich87}, oblivious transfers~\cite{DBLP:journals/iacr/Rabin05,DBLP:conf/istcs/IshaiK97},  and secret-sharing-based techniques~\cite{DBLP:journals/cacm/Shamir79,DBLP:conf/eurocrypt/GilboaI14,DBLP:conf/eurocrypt/BoyleGI15,DBLP:conf/tcc/KomargodskiZ16,DBLP:conf/ccs/DolevGL15,DBLP:conf/nsdi/WangYGVZ17,DBLP:journals/isci/EmekciMAA14,DBLP:conf/fc/LueksG15} are not prone to access-pattern attacks.

\item\emph{Workload-skew attack}: An adversary, having the knowledge of frequent selection queries by observing many queries, can estimate which encrypted tuples potentially satisfy the frequent selection queries. Except access-pattern-hiding techniques, all the cryptographic techniques are prone to workload-skew attack (we will discuss workload-skew attack in detail in \S\ref{subsec:Handling Workload-skew Attack}).
\end{enumerate}

To the best of our knowledge, there is no cryptographic technique that prevents all the four attacks on a \emph{skewed dataset}. 

\end{itemize}


\smallskip
\noindent\textbf{Our contribution.} While the race to develop cryptographic solutions that (\textit{i}) are efficient, (\textit{ii}) support complex SQL queries, (\textit{iii}) offer provable security from the application's perspective is ongoing, this paper departs from the above well-trodden path by exploring a \emph{radically different (but complementary) approach to secure data processing in the cloud}. Our approach is intended for situations when only part of the data is sensitive, while the remainder (that may consist of the majority) is non-sensitive.

In particular, we propose a {\bf partitioned computation model} that exploits such a classification of data into sensitive/non-sensitive subsets to develop efficient data processing solutions with {\bf provable security guarantees}. In partitioned computing, sensitive data is outsourced in an appropriate encrypted form, while non-sensitive data can be outsourced in cleartext form. Partitioned computing, potentially, provides significant benefits by (\textit{i}) avoiding (expensive) cryptographic operations on non-sensitive data, and (\textit{ii}) allowing query processing on non-sensitive data to exploit indices. Such indices (that cannot be easily supported alongside encryption-based mechanisms in a non-interactive setting) are a key mechanism for efficient query processing in traditional database systems.\footnote{The sensitive and non-sensitive data classification, which is common in industries for secure computing~\cite{url2,url3} and done via appropriately using existing techniques surveyed in~\cite{DBLP:journals/sigkdd/FarkasJ02}; for example, (\textit{i}) inference detection using graph-based semantic data modeling~\cite{DBLP:conf/sp/Hinke88}, (\textit{ii}) user-defined relationships between sensitive and non-sensitive data~\cite{smith1990modeling}, (\textit{iii}) constraints-based mechanisms, (\textit{iv}) sensitive patterns hiding using sanitization matrix~\cite{lee2004hiding}, and (\textit{v}) common knowledge-based association rules~\cite{DBLP:conf/dasfaa/LiSY07}. However, it is important to mention here that non-sensitive data can, over time, become sensitive and/or lead to inferences about sensitive data. This is an inevitable risk of the approaches that exploit sensitive data classification. Note that all the above-mentioned work based on sensitive/non-sensitive classification make a similar assumption.}

While partitioned computing offers new opportunities for efficient and secure data processing on the cloud, it raises a new security challenge (\S\ref{sec:Preliminaries_and_Problem_Statements}) -- of leakage, due to the joint processing of the encrypted (sensitive) dataset and of the plaintext (non-sensitive) datasets. We refer to this security challenge as {\bf partitioned data security} challenge. Our work will formalize the security definition that will drive the development of the proposed prototype, entitled \textsc{Panda} (to refer to \emph{PA}rtitio\emph{N D}at\emph{A}). In \textsc{Panda}, we develop a query processing technique, entitled \emph{query binning} (QB) to prevent leakage of sensitive data due to the simultaneous execution of sensitive and non-sensitive data. In addition, \textsc{Panda} extends QB to answer   selection, range, and join queries. We will, also, show two interesting effects of using \textsc{Panda}:
\begin{enumerate}[nolistsep,noitemsep,leftmargin=0.01in]
  \item By avoiding cryptographic processing on non-sensitive data, \emph{the joint cost of communication and computation of \textsc{Panda}'s QB is significantly less than the computation cost of a strongly secure cryptographic technique}\footnote{QB trades off increased communication costs for executing queries, while reducing very significantly cryptographic operations. This tradeoff significantly improves performance, especially, when using cryptographic mechanisms, \textit{e}.\textit{g}., fully homomorphic encryption that takes several seconds to compute a single operation~\cite{DBLP:journals/csur/MartinsSM17}, secret-sharing-based techniques that take a few seconds~\cite{DBLP:journals/isci/EmekciMAA14},
or techniques such as bilinear maps that take over 1.5 hours to perform joins on a dataset of size less than 10MB~\cite{DBLP:journals/tods/PangD14}. When considering such cryptography, increased communication overheads are fully compensated by the savings. A similar observation, albeit in a very different context was also observed in~\cite{DBLP:conf/sigmod/OktayMKK15} in the context of MapReduce, where overshuffling to prevent the adversary to infer sensitive keys in the context of hybrid cloud was shown to be significantly better compared to private side operations.} (\textit{e}.\textit{g}., homomorphic encryptions, or secret-sharing-based technique~\cite{DBLP:conf/dbsec/DolevL016,DBLP:journals/pvldb/GuptaLMP0A19} that hides access-patterns --- the identity of the tuple satisfying the query) on the entire encrypted data; and hence, QB improves the performance of strong cryptographic techniques over a large-scale dataset (\S\ref{sec:Experimental Validation}).

\item \textsc{Panda}'s QB  provides \emph{enhanced security by preventing several attacks such as output size, frequency-count, and workload-skew attacks}, even when the underlying cryptographic technique is susceptible to such attacks (\S\ref{subsec:Enhancing security-levels of indexable techniques}).

\end{enumerate}



\smallskip
\noindent\textbf{Outline.} The primary contributions of this paper and its online are as follows:
\begin{enumerate}[nolistsep,noitemsep,leftmargin=0.01in]
\item The partition computation model and inference attack due to the joint processing over sensitive and non-sensitive data (\S\ref{sec:Preliminaries_and_Problem_Statements}).
  \item A formal definition of \emph{partitioned data security} when jointly processing sensitive and non-sensitive data (\S\ref{subsec:Security Definition and Correctness}).
  \item  An efficient QB approach (\S\ref{sec:Query Bucketization}) that guarantees partitioned data security, supporting cloud-side-indexes, and that can be built on top of any cryptographic technique.

 \item Methods to deal with join queries, range queries, insert operation, and workload-skew attacks (\S\ref{sec:Other Operations}).
  \item  A weak cryptographic technique (\textit{e}.\textit{g}., cloud-side indexable techniques~\cite{DBLP:conf/dbsec/ShmueliWEG05,arx-popa-2017,DBLP:conf/sdmw/EloviciWSG04}) becomes secure and efficient when mixed with QB (\S\ref{subsec:Enhancing security-levels of indexable techniques}).
\item Experimental evaluation of \textsc{Panda} under different settings and queries (\S\ref{sec:Experimental Validation}).
\end{enumerate}

\smallskip
\noindent\textbf{Conference version.} A preliminary version of this paper was accepted and presented in IEEE ICDE~\cite{DBLP:conf/icde/Mehrotra0UM19}. The conference version includes the following additional concept, which is not provided in this version, due to space restriction: an analytical model to show when QB works better compared to a pure cryptographic technique (\S~V.A of~\cite{DBLP:conf/icde/Mehrotra0UM19}).}


}

\bgroup
\def\arraystretch{1.1}
\begin{table}[h]
\centering
\centering
\begin{tabular}{|l|p{12cm}|}
\hline \textbf{Notations} & \textbf{Meaning} \\ \hline\hline
$|S|$ & Number of sensitive data values \\\hline
$|\mathit{NS}|$ & Number of non-sensitive data values\\\hline
$R_s$ & Sensitive parts of a relation $R$\\\hline
$R_{\mathit{ns}}$ & Non-sensitive parts of a relation $R$\\\hline
$s_i$ and $\mathit{ns}_j$  & $i^{\mathit{th}}$ sensitive and $j^{\mathit{th}}$ non-sensitive values \\\hline
$\mathit{SB}$ & The number of sensitive bins \\\hline
$\mathit{SB}_i$ & $i^{\mathit{th}}$ sensitive bin \\\hline
$|\mathit{SB}|=y$ & Sensitive values in a sensitive bin or the size of a sensitive bin\\\hline
$\mathit{NSB}$ & The number of non-sensitive bins \\\hline
$\mathit{NSB}_i$ & $i^{\mathit{th}}$ non-sensitive bin \\\hline
$|\mathit{NSB}|=x$ & Non-sensitive values in a non-sensitive bin or the size of a non-sensitive bin\\\hline
$q(w)$ & A query, $q$, for a predicate $w$\\\hline
$q(W_{\mathit{ns}})(R_{\mathit{ns}})$ & A query, $q$, for a set, $W_{\mathit{ns}}$, of predicates in cleartext over $R_{\mathit{ns}}$\\\hline
$q(W_s)(R_s)$ & A query, $q$, for a set, $W_s$, of predicates in encrypted form over $R_s$\\\hline
$q(W)(R_s,R_{\mathit{ns}})[A]$ & A query, $q$, for a set, $W$, of values, searching on the attribute, $A$, of the relations $R_s$ and $R_{\mathit{ns}}$, where $W=W_s \cup W_{\mathit{ns}} $ \\\hline
$E(t_i)$ & $i^{\mathit{th}}$ encrypted tuple\\\hline
\end{tabular}
\caption{Notations used in the paper.}
\label{tab:notations}
\end{table}
\egroup

\section{Partitioned Computation}
\label{sec:Preliminaries_and_Problem_Statements}
In this section, we first define more precisely what we mean by partitioned computing, illustrate how such a computation can leak information due to the joint processing of sensitive and non-sensitive data, discuss the corresponding security definition, and finally, discuss the system and adversarial models under which we will develop our solutions. Table~\ref{tab:notations} enlists notations used in this paper.

\subsection{The Partitioned Computation Model}
We assume the following two entities in our model:
\begin{enumerate}[noitemsep,leftmargin=0.01in]
  \item \emph{\textbf{A trusted database (DB) owner}} who divides a relation $R$ having attributes, say $A_1, A_2, \ldots, A_n$, into the following two relations based on row-level data sensitivity: $R_s$ and $R_{\mathit{ns}}$ containing all sensitive and non-sensitive tuples, respectively. 
      The DB owner outsources the relation $R_{\mathit{ns}}$ to a public cloud. The tuples of the relation $R_s$ are encrypted using any existing non-deterministic encryption~\cite{DBLP:journals/jcss/GoldwasserM84} mechanism before outsourcing to the same public cloud.

       {In our setting, the DB owner has to store metadata such as searchable values and their frequency counts, which will be used for appropriate query formulation using the proposed query binning (QB) algorithm, (on receiving a query from a user). The size of metadata is smaller than the size of the original data.
      The DB owner is assumed to have sufficient storage for such metadata, and also computational capabilities to execute QB algorithm, encryption (of queries keywords) and decryption (of the results).}

      { {\textbf{Note.} The tasks at the DB owner for metadata data storage and QB algorithm execution (which requires to execute Algorithm~\ref{alg:bin_creation} and Algorithm~\ref{alg:bin_retrieval}, will be explained in \S\ref{sec:Query Bucketization}) could, potentially, be executed at the cloud, if the cloud supports a trusted hardware, \textit{e}.\textit{g}., SGX. However, using SGX for QB is nontrivial, since, now, the entire dataset needs to be encrypted and send to SGX that will decrypt and execute QB. In contrast, the task of rewriting the queries (sent by the DB owner) based on bin information, \textit{i}.\textit{e}., Algorithm~\ref{alg:bin_retrieval},  is relatively simple and can be done at SGX hosted by the cloud.}}

      { {In addition, proxy reencryption in our setting is complex, since the results/answers to the query may have additional outputs, which will need to be filtered out at the trusted side. Thus, solutions based on proxy reencryption will not work in our settings.
      }}


  \item \emph{\textbf{The untrusted public cloud}} that stores the databases, executes queries, and provides answers to the DB owner.
\end{enumerate}

\smallskip
\noindent
\textbf{Query execution.} Let us consider a query $q$ over the relation $R$, denoted by $q(R)$. A partitioned computation strategy splits the execution of $q$ into two independent subqueries: $q(R_s)$: a query to be executed on the encrypted sensitive relation $R_s$, and $q(R_{\mathit{ns}})$: a query to be executed on the non-sensitive relation $R_{\mathit{ns}}$. The final result is computed (using a query $q_{\mathit{merge}}$) by appropriately merging the results of the two subqueries at the DB owner side. In particular, the query $q$ on a relation $R$ is partitioned, as follows:
$$q(R) = q_{\mathit{merge}}\Big(q(R_s), q(R_{\mathit{ns}}) \Big)$$
\noindent Let us illustrate partitioned computations through an example.

\begin{figure}[h]
\centering
\begin{tabular}{|l||l|l|l|l|l|l|l|}
    \hline
      &\textbf{EId} & \textbf{FirstName} & \textbf{LastName} & \textbf{SSN}  & \textbf{Office\#} & \textbf{Department} \\ \hline\hline
    $t_1$& E101 & Adam & Smith       & 111&1 & Defense  \\ \hline
    $t_2$& E259 & John & Williams   & 222&2 & Design   \\ \hline
    $t_3$& E199 & Eve  & Smith      & 333&2  & Design   \\ \hline
    $t_4$& E259 & John & Williams   & 222&6 & Defense   \\ \hline
    $t_5$& E152 & Clark & Cook   & 444&1 & Defense   \\ \hline
    $t_6$& E254 & David & Watts   & 555&4 & Design   \\ \hline
    $t_7$& E159 & Lisa & Ross   & 666&2 & Defense   \\ \hline
    $t_8$& E152 & Clark & Cook   & 444&3 & Design   \\ \hline

  \end{tabular}
  \caption{A relation: \textit{Employee}.}
  \label{fig:employee relation}
\end{figure}

\noindent\textbf{Example 1.} Consider an \emph{Employee} relation, see Figure~\ref{fig:employee relation}. Note that the notation $t_i$ ($1\leq i\leq 8$) is not an attribute of the relation; we used this to indicate the $i^{\mathit{th}}$ tuple. In this relation, the attribute {\em SSN} is sensitive, and furthermore, all tuples of employees for the {\em Department} $=$ ``\texttt{Defense}'' are sensitive. In such a case, the \emph{Employee} relation may be stored as the following three relations: (\textit{i}) \emph{Employee1} with attributes {\em EId} and {\em SSN} (see Figure~\ref{fig:employee1 relation}); (\textit{ii}) \emph{Employee2} with attributes {\em EId}, {\em FirstName}, {\em LastName}, {\em Office\#}, and {\em Department}, where {\em Department} $=$ ``\texttt{Defense}'' (see Figure~\ref{fig:employee2 relation}); and (\textit{iii}) \emph{Employee3} with attributes {\em EId}, {\em FirstName}, {\em LastName}, {\em Office\#}, and {\em Department}, where {\em Department} $<>$ ``\texttt{Defense}'' (see Figure~\ref{fig:employee3 relation}). Since the relations \emph{Employee1} and \emph{Employee2} (Figures~\ref{fig:employee1 relation} and~\ref{fig:employee2 relation}) contain only sensitive data, these two relations are encrypted before outsourcing, while \emph{Employee3} (Figure~\ref{fig:employee3 relation}), which contains only non-sensitive data, is outsourced in cleartext. We assume that the sensitive data is strongly encrypted such that the property of \emph{ciphertext indistinguishability} (\textit{i}.\textit{e}., an adversary cannot distinguish pairs of ciphertexts) is achieved. Thus, the two occurrences of \texttt{E152} have two different ciphertexts.

\begin{figure}[h]
\begin{center}
  \begin{minipage}[t]{.35\linewidth}
  \centering
\begin{tabular}{|l|l|}\hline
          \textbf{EId}  & \textbf{SSN}   \\\hline\hline
    E101 & 111   \\\hline
    E259 & 222   \\\hline
    E199 & 333   \\\hline
     E152 & 444   \\\hline
     E254 & 555   \\\hline
     E159 & 666   \\\hline
  \end{tabular}
  \subcaption{A sensitive relation: \textit{Employee1}.}
    \label{fig:employee1 relation}
  \end{minipage}
    \begin{minipage}[t]{.47\linewidth}
  \centering
\centering
\begin{tabular}{|l||l|l|l|l|l|l|}\hline
         & \textbf{EId}  & \textbf{FirstName} & \textbf{LastName} &  \textbf{Office\#} & \textbf{Department} \\ \hline\hline
    $t_1$& E101 & Adam & Smith       & 1 & Defense  \\ \hline
    $t_4$& E259 & John & Williams   & 6 & Defense   \\ \hline
    $t_5$& E152 & Clark & Cook   &1 & Defense   \\ \hline
    $t_7$& E159 & Lisa & Ross   &2 & Defense   \\ \hline
  \end{tabular}
  \subcaption{A sensitive relation: \textit{Employee2}.}
\label{fig:employee2 relation}
  \end{minipage}

  \begin{minipage}[t]{.47\linewidth}
  \centering
\centering
\begin{tabular}{|l||l|l|l|l|l|l|}
    \hline
    & \textbf{EId}  & \textbf{FirstName} & \textbf{LastName} &  \textbf{Office\#} & \textbf{Department} \\ \hline\hline
    $t_2$& E259 & John & Williams   &2 & Design   \\ \hline
    $t_3$& E199 & Eve  & Smith      &2  & Design   \\ \hline
    $t_6$& E254 & David & Watts   & 4 & Design   \\ \hline
    $t_8$& E152 & Clark & Cook   & 3 & Design   \\ \hline
  \end{tabular}
  \subcaption{A non-sensitive relation: \textit{Employee3}.}
\label{fig:employee3 relation}
  \end{minipage}
\end{center}
\caption{Three relations obtained from \textit{Employee} relation.}
\label{fig:three relations}
\end{figure}

Consider a query \texttt{q: SELECT FirstName, LastName, Office\#, Department from Employee where FirstName = John}. In the partitioned computation, the query \texttt{q} is partitioned into two subqueries: $q_s$ that executes on \texttt{Employee2}, and $q_{\mathit{ns}}$ that executes on \texttt{Employee3}. $q_s$ will retrieve the tuple $t_4$ while $q_{\mathit{ns}}$ will retrieve the tuple $t_2$. $q_{merge}$ in this example is simply a union operator. Note that the execution of the query \texttt{q} will also retrieve the same tuples.

However, such a partitioned computation, if performed naively, leads to inferences about sensitive data from non-sensitive data. Before discussing the inference attacks, we first present the adversarial model.

\subsection{Adversarial Model}
\label{subsec:Adversarial Model}

We assume an honest-but-curious adversary that is \textit{not trustworthy}. The honest-but-curious adversary is considered widely in the standard database-as-a-service query processing model, keyword searches, and join processing~\cite{DBLP:conf/stoc/CanettiFGN96,DBLP:conf/infocom/YuWRL10,DBLP:conf/icdcs/WangCLRL10,DBLP:conf/ccs/YuWRL10,DBLP:journals/tc/WangCWRL13}. An honest-but-curious adversarial public cloud stores an outsourced dataset without tampering, correctly computes assigned tasks, and returns answers; however, it may exploit side knowledge (\textit{e}.\textit{g}., query execution, background knowledge, and the output size) to gain as much information as possible about the sensitive data.\footnote{The honest-but-curious adversary cannot launch any attack against the DB owner. We do not consider cyber-attacks that can exfiltrate data from the DB owner directly, since defending against generic cyber-attacks is outside the scope of this paper.} Furthermore, the honest-but-curious adversary can eavesdrop on the communication channels between the cloud and the DB owner, and that may help in gaining knowledge about sensitive data, queries, or results; hence, a secure channel is assumed. In our setting, the adversary has full access to the following:

\begin{enumerate}[noitemsep,leftmargin=0.01in]
  \item All the non-sensitive data. For example, for the Employee relation in Example 1, an adversary knows the complete \emph{Employee3} relation (refer to Figure~\ref{fig:employee3 relation}).
  \item \emph{Auxiliary/background} information of the sensitive data. The auxiliary information may contain metadata, schema of the relation, and the number of tuples in the relation (note that having an adversary with the auxiliary information is also considered in literature~\cite{DBLP:conf/ccs/NaveedKW15,DBLP:conf/ccs/KellarisKNO16}). In Example 1, the adversary knows that there are two sensitive relations, one of them containing six tuples and the other one containing four tuples, in the \emph{Employee1} and the \emph{Employee2} relations; Figures~\ref{fig:employee1 relation} and~\ref{fig:employee2 relation}. In contrast, the adversary is not aware of the following information before the query execution: how many people work in a specific sensitive department, is a specific person working only in a sensitive department, only in a non-sensitive department, or both.

  \item Adversarial view. When executing a query, an adversary knows which encrypted sensitive tuples and cleartext non-sensitive tuples are sent in response to a query. We refer this as the adversarial view, denoted by $\mathit{AV}$: $\mathit{AV} = \mathit{In}_c \cup \mathit{Op}_c$, where $\mathit{In}_c$ refers to the query arrives at the cloud and $\mathit{Op}_c$ refers to the encrypted and non-encrypted tuples, transmitted in response to $\mathit{In}_c$. For example, the first row of Table~\ref{tab:answer table} shows an adversarial view that shows that $\mathit{Op}_c= t_2$ tuples from the non-sensitive relation and encrypted $\mathit{Op}_c=t_4$ tuples from the sensitive relation are returned to answer the query for $\mathit{In}_c=$ \texttt{E259}.

  \item Some frequent query values. The adversary observes query predicates on the non-sensitive data, and hence, can deduce the most frequent query predicates by observing many queries.
\end{enumerate}

\subsection{Inference Attacks in Partitioned Computations}
To see the inference attack on the sensitive data while jointly processing sensitive and non-sensitive data, consider the following three queries on the \emph{Employee2} and \emph{Employee3} relations; refer to Figures~\ref{fig:employee2 relation} and~\ref{fig:employee3 relation}.

\smallskip
\noindent\textbf{Example 2.} (\textit{i}) retrieve tuples corresponding to employee \texttt{E259}, (\textit{ii}) retrieve tuples corresponding to employee \texttt{E101}, and (\textit{iii}) retrieve tuples corresponding to employee \texttt{E199}.\footnote{We used random \emph{Eids}, which is also common in a real employee relation. In contrast, in sequential ids, the absence of an id from the non-sensitive relation directly informs the adversary that the given id exists in the sensitive relation.} When answering a query, the adversary knows the tuple ids of retrieved encrypted tuples and the full information of the returned non-sensitive tuples. We refer to this information gain by the adversary as the \emph{adversarial view}, shown in Table~\ref{tab:answer table}, where $\mathit{E(t_i)}$ denotes an encrypted tuple $t_i$.

\begin{table}[h]
  \centering
    \begin{tabular}{|l|l|l|}
    \hline
    \textbf{Query value} & \multicolumn{2}{|c|}{\textbf{Returned tuples/Adversarial view}}           \\ \hline
    ~                          & \textbf{Employee2} & \textbf{Employee3} \\ \hline\hline
    E259                      & $\mathit{E(t_4)}$        & $t_2$        \\ \hline
    E101                              & $\mathit{E(t_1)}$        & null      \\ \hline
    E199                            & null        & $t_3$      \\ \hline
    \end{tabular}
    \caption{Queries and returned tuples/adversarial view.}
    \label{tab:answer table}
\end{table}
Outputs of the above three queries will reveal enough information to learn something about sensitive data. In the first query, the adversary learns that \texttt{E259} works in both sensitive and non-sensitive departments, because the answers obtained from the two relations contribute to the final answer. Moreover, the adversary may learn which sensitive tuple has an \emph{Eid} equals to \texttt{E259}. In the second query, the adversary learns that \texttt{E101} works only in a sensitive department, because the query will not return any answer from the Employee3 relation. In the third query, the adversary learns that \texttt{E199} works only in a non-sensitive department.

\subsection{The Query Binning (QB) Approach: An Overview}
In order to prevent the inference attack in the partitioned computation, we need a new security definition. Before we discuss the formal definition of partitioned data security (\S\ref{subsec:Security Definition and Correctness}), we first provide a possible solution to prevent inference attacks and then intuition for the security definition.

The query binning (QB) strategy stores a non-sensitive relation, say $R_{\mathit{ns}}$, in cleartext while it stores a sensitive relation, say $R_s$, using a cryptographically secure approach. QB prevents leakage such as in Example 2 by appropriately mapping a query for a predicate, say $q(w)$, to corresponding queries both over the non-sensitive relation, say $q(W_{\mathit{ns}})(R_{\mathit{ns}})$, and encrypted relation, say $q(W_s)(R_s)$. The queries $q(W_{\mathit{ns}})(R_{\mathit{ns}})$ and $q(W_s)(R_s)$, each represents a set of predicates (or selection queries) that are executed over the relation $R_{\mathit{ns}}$ in plaintext and, respectively, over the sensitive relation $R_s$, using the underlying cryptographic method. The set of predicates in $q(W_{\mathit{ns}})(R_\mathit{ns})$ (likewise in $q(W_S)(R_s)$) correspond to the non-sensitive (sensitive) {\em bins} including the predicate $w$, denoted by $\mathit{NSB}$ ($\mathit{SB}$). The predicates in $q(W_s)(R_s)$ are encrypted before transmitting to the cloud.

The bins are selected such that: (\textit{i}) $w \in q(W_{\mathit{ns}})(R_{\mathit{ns}}) \cap q(W_s)(R_s)$ to ensure that all the tuples containing the predicate $w$ are retrieved, and, (\textit{ii}) joint execution of the queries $q(W_{\mathit{ns}})(R_\mathit{ns})$ and $q(W_s)(R_s)$ (hereafter, denoted by $q(W)(R_s,R_\mathit{ns})$, where $W=W_s \cup W_{\mathit{ns}}$) does not leak the predicate $w$. Results from the execution of the queries $q(W_{\mathit{ns}})(R_{\mathit{ns}})$ and $q(W_s)(R_s)$ are decrypted, possibly filtered, and merged to generate the final answer. Note that \emph{bins are created only once for all the values of a searching attribute before any query is executed}. The details of the bin formation will be discussed in \S\ref{sec:Query Bucketization}.


For answering the above-mentioned three queries, QB creates two bins on sensitive parts: $\{$\texttt{E101}, \texttt{E259}$\}$, $\{$\texttt{E152}, \texttt{E159}$\}$, and two sets on non-sensitive parts: $\{$\texttt{E259}, \texttt{E254}$\}$, $\{$\texttt{E199}, \texttt{E152}$\}$. Table~\ref{tab:answer table qb} illustrates the generated adversarial view when QB is used to answer queries as shown in Example 2. In this example, row 1 of Table~\ref{tab:answer table qb} shows that this instance of QB maps the query for \texttt{E259} to  $\langle$\texttt{E259}, \texttt{E254}$\rangle$ over cleartext and to encrypted version of values for $\langle$\texttt{E259}, \texttt{E101}$\rangle$ over sensitive data. Note that simply from the generated adversarial views, the adversary  cannot determine the query value $w$ (\texttt{E259} in the example) or find a value that is shared between the two sets. Thus, while answering a query, the adversary cannot learn which employee works only in defense, design, or in both.

The reason is that the desired query value, $w$, is encrypted with other encrypted values of $W_s$, and, furthermore, the query value, $w$, cannot be distinguished from many requested non-sensitive values of $W_{\mathit{ns}}$, which are in cleartext. Consequently, \emph{the adversary is unable to find an intersection of the two sets, which is the exact value}.\footnote{For hiding an exact selection predicate over an encrypted relation regardless of data sensitivity, an approach to create a set of selection predicates including the exact predicate is presented in~\cite{DBLP:journals/isci/LiuZWT14}, which, however, cannot be used to search over sensitive and non-sensitive relations or multiple relations, due to not dealing with inference attacks.}

\begin{table}[h]
  \centering
    \begin{tabular}{|l|l|l|}
    \hline
    \textbf{Query value} & \multicolumn{2}{|c|}{\textbf{Returned tuples/Adversarial view}}           \\ \hline
    ~                         & \textbf{Employee2} & \textbf{Employee3} \\ \hline\hline
    E259                      & $\mathit{E(t_4)}$, $\mathit{E(t_1)}$       & $t_2$, $t_6$  \\ \hline
    E101                      & $\mathit{E(t_4)}$, $\mathit{E(t_1)}$       & $t_3$, $t_8$  \\ \hline
    E199                      & $\mathit{E(t_4)}$, $\mathit{E(t_1)}$       & $t_3$, $t_8$  \\ \hline
    \end{tabular}
    \caption{Queries and returned tuples/adversarial view, following QB.}
    \label{tab:answer table qb}
\end{table}
Thus, in a joint processing of sensitive and non-sensitive data, \emph{the goal of the adversary is to find as much sensitive information as possible (using the adversarial view or background knowledge)}, and \emph{the goal of a secure technique is to prevent information leakage through the joint processing of non-sensitive and sensitive data}.


\section{Partitioned Data Security}
\label{subsec:Security Definition and Correctness}

In this section, we formalize the notion of {\em partitioned data security} that establishes when a partitioned computation over sensitive and non-sensitive data does not leak any sensitive information. Note that an adversary may seek to infer sensitive information using the adversarial view created during query processing, knowledge of output size, frequency counts, and workload characteristics. We begin by first formalizing the concepts of: \emph{associated values}, \emph{associated tuples}, and \emph{relationship between counts of sensitive values}.\footnote{To develop the notation, defining security, and developing QB (\S\ref{sec:Query Bucketization}), we assume that search is performed on a specific attribute, $A$, over a relation, $R$. The approach trivially generalizes when several attributes are searchable -- we need to maintain metadata required for QB not just for $A$, but for all searchable attributes in $R$.}

\parskip 0pt
\setlength{\parindent}{15pt}

\medskip
\noindent\emph{Notations used in the definitions}. Let $t_1, t_2, \ldots, t_m$ be tuples of a sensitive relation, say $R_s$. Thus, the relation $R_s$ stores the encrypted tuples $E(t_1), E(t_2), \ldots, E(t_m)$. Let $s_1,s_2,\ldots, s_{m^{\prime}}$ be values of an attribute, say $A$, that appears in one of the sensitive tuples of $R_s$. Note that $m^{\prime} \leq m$, since several tuples may have an identical value. Furthermore, $s_i \in \mathit{Domain}(A)$, $i=1,2, \dots, m^{\prime}$, where $\mathit{Domain}(A)$ represents the domain of values the attribute $A$ can take. By $\#_s(s_i)$, we refer to the number of sensitive tuples that have $s_i$ as the value for attribute $A$. We further define $\#_s(v) = 0, \forall v \in \mathit{Domain}(A)$, $v \notin s_1,s_2,\ldots, s_{m^{\prime}}$. Let $t_1, t_2, \ldots, t_n$ be tuples of a non-sensitive relation, say $R_{\mathit{ns}}$. Let $\mathit{ns}_1,\mathit{ns}_2,\ldots, \mathit{ns}_{n^{\prime}}$ be values of the attribute $A$ that appears in one of the non-sensitive tuples of $R_{\mathit{ns}}$. In analogy with the case where the relation is sensitive, $n^{\prime} \leq n$, and $\mathit{ns}_i \in \mathit{Domain}(A)$, $i=1,2, \dots, n^{\prime}$.

\medskip
\noindent\emph{Associated values}. Let $e_i = E(t_i)[A]$ be the encrypted representation of an attribute value of $A$ in a sensitive tuple of the relation $R_s$, and $\mathit{ns}_j$ be a value of the attribute $A$ for some tuple of the relation $R_{\mathit{ns}}$. We say that $e_i$ is \emph{associated} with $\mathit{ns}_j$, (denoted by $\overset{\mathrm{a}}{=}$), if the plaintext value of $e_i$ is identical to the value $\mathit{ns}_j$. In Example 1, the value of the attribute \texttt{Eid} in tuple $t_4$ (of \emph{Employee2}, see Figure~\ref{fig:employee2 relation}) is associated with the value of the attribute \texttt{Eid} in tuple $t_2$ (of \emph{Employee3}, see Figure~\ref{fig:employee3 relation}), since both values correspond to \texttt{E259}.

\medskip
\noindent\emph{Associated tuples}. Let $t_i$ be a sensitive tuple of the relation $R_s$ (\textit{i}.\textit{e}., $R_s$ stores encrypted representation of $t_i$) and $t_j$ be a non-sensitive tuple of the relation $R_{\mathit{ns}}$. We state that $t_i$ is associated with $t_j$ (for an attribute, say $A$) iff the value of the attribute $A$ in $t_i$ is associated with the value of the attribute $A$ in $t_j$ (\textit{i}.\textit{e}., $t_i[A] \overset{\mathrm{a}}{=} t_j[A]$). Note that this is the same as stating that the two values of attribute $A$ are equal for both tuples.

\medskip
\noindent\emph{Relationship between counts of sensitive values}. Let $v_i$ and $v_j$ be two distinct values in $\mathit{Domain}(A)$. We denote the relationship between the counts of sensitive tuples with these $A$ values (\textit{i}.\textit{e}., $\#_s(v_i)$ (or $\#_s(v_j)$)) by $v_i \overset{\mathrm{r}}{\sim} v_j$. Note that $\overset{\mathrm{r}}{\sim}$ can be one of $<, =$, or $>$ relationships. For instance, in Example 1, the \texttt{E101} $\overset{\mathrm{r}}{\sim}$ \texttt{E259} corresponds to $=$, since both values have exactly one sensitive tuple (see Figure~\ref{fig:employee2 relation}), while \texttt{E101} $\overset{\mathrm{r}}{\sim}$ \texttt{E199} is $>$, since there is one sensitive tuple with value \texttt{E101} while there is no sensitive tuple with \texttt{E199}.

\medskip
Given the above definitions, we can now formally state the security requirement that ensures that simultaneous execution of queries over sensitive (encrypted) and non-sensitive (plaintext) data does not leak any information. Before that, we wish to mention the need of a new security definition in our context. The inference attack in the partitioned computing can be considered to be related to the known-plaintext attack (KPA) wherein the adversary knows some plaintext data which is hidden in a set of ciphertext. In KPA, the adversary's goal is to determine which ciphertext data is related to a given plaintext, \textit{i}.\textit{e}., determining a mapping between ciphertext and the corresponding plaintext data representing the same value. In our setup, non-sensitive values are visible to the adversary in plaintext. However, the attacks are different since, unlike the case of KPA, in our setup, the ciphertext data might not contain any data value that is the same as some non-sensitive data visible to the adversary in plaintext.\footnote{The HBC adversary cannot launch the chosen-plaintext attack (CPA) and the chosen-ciphertext attack (CCA). Since the sensitive data is non-deterministically encrypted (by our assumption), it is not prone to the ciphertext only attack (COA).}

\medskip
\noindent\textbf{Definition: Partitioned Data Security.} Let $R$ be a relation containing sensitive and non-sensitive tuples. Let $R_s$ and $R_{\mathit{ns}}$ be the sensitive and non-sensitive relations, respectively. Let $\mathit{AV}$ be an adversarial view generated for a query $q(w)(R_s,R_{\mathit{ns}})[A]$, where the query, $q$, for a value $w$ in the attribute $A$ of the $R_s$ and $R_{\mathit{ns}}$ relations. Let $X$ be the auxiliary information about the sensitive data, and $\mathit{Pr_{Adv}}$ be the probability of the adversary knowing any information. A query execution mechanism ensures the partitioned data security if the following two properties hold:
\begin{enumerate}[leftmargin=0.01in]


\item \noindent $\mathit{Pr}_{\mathit{Adv}}[e_i \overset{\mathrm{a}}{=} \mathit{ns}_j|X] = \mathit{Pr}_{\mathit{Adv}}[e_i \overset{\mathrm{a}}{=} \mathit{ns}_j|X, \mathit{AV}]$, where $e_i= E(t_i)[A]$ is the encrypted representation for the attribute value $A$ for any tuple $t_i$ of the relation $R_s$ and $\mathit{ns}_j$ is a value for the attribute $A$ for any tuple of the relation $R_{\mathit{ns}}$.



\item $\mathit{Pr}_{\mathit{Adv}}[v_i \overset{\mathrm{r}}{\sim} v_j|X] = \mathit{Pr}_{\mathit{Adv}}[v_i \overset{\mathrm{r}}{\sim} v_j | X, \mathit{AV}]$, for all $v_i, v_i \in \mathit{Domain}(A)$.

\end{enumerate}


\medskip
 The first equation (1) captures the fact that an initial probability of \emph{associating} a sensitive tuple with a non-sensitive tuple will be identical after executing a query on the relations. Thus, an adversary cannot learn anything from an adversarial view generated after the query execution. Satisfying this condition also prevents us in achieving success against KPA. The second equation (2) states that the probability of an adversary gaining information about the relative frequency of sensitive values does not increase after the query execution. In Example 2, an execution of any three queries (for values \texttt{E101}, \texttt{E199}, or \texttt{E259}) without using QB does not satisfy the above first equation. For example, the query for \texttt{E199} retrieves the only tuple from non-sensitive relation, and that changes the probability of estimating whether \texttt{E199} is sensitive or non-sensitive to 0 as compared to an initial probability of the same estimation, which was 1/4. Hence, execution of the three queries violates partitioned data security. However, the query execution for \texttt{E259} and \texttt{E101} satisfies the second equation, since the count of returned tuples from \emph{Employee2} is equal. Hence, the adversary cannot distinguish between the count of the values (\texttt{E259} and \texttt{E101}) in the domain of \texttt{Eid} of \emph{Employee2} relation.

%
%

\section{Query Binning Technique}
\label{sec:Query Bucketization}
We develop our strategy initially under the assumption that queries are only on a single attribute, say $A$. QB approach takes as inputs: (\textit{i}) the set of data values (of the attribute $A$) that are sensitive, along with their counts, and (\textit{ii}) the set of data values (of the attribute $A$) that are non-sensitive, along with their counts. QB returns a partition of attribute values that form the query bins for both the sensitive as well as for the non-sensitive parts of the query. We begin in \S\ref{subsec:The Base Case} by developing the approach for the case when a sensitive tuple is associated with at most one non-sensitive tuple (Algorithm~\ref{alg:bin_creation}). 

\parskip 0pt
\setlength{\parindent}{15pt}

Informally, QB distributes attribute values in a matrix, where rows are sensitive bins, and columns are non-sensitive bins. For example, suppose there are 16 values, say $0, 1,\ldots,15$, and assume all the values have sensitive and associated non-sensitive tuples. Now, the DB owner arranges 16 values in a $4\times4$ matrix, as follows:

\begin{center}
\begin{tabular}{|l||l|l|l|l|}
  \hline
  & $\mathit{NSB}_0$ & $\mathit{NSB}_1$ & $\mathit{NSB}_2$ & $\mathit{NSB}_3$  \\  \hline  \hline
    $\mathit{SB}_0$ & 11  & 2 & 5 & 14 \\  \hline
    $\mathit{SB}_1$ &  10&  3&  8&7  \\  \hline
    $\mathit{SB}_2$ &  0&  15& 6 & 4 \\  \hline
    $\mathit{SB}_3$ &  13&  1& 12 &9  \\  \hline
    \end{tabular}
\end{center}

In this example, we have four sensitive bins: $\mathit{SB}_0$ \{11,2,5,14\}, $\mathit{SB}_1$ \{10,3,8,7\}, $\mathit{SB}_2$ \{0,15,6,4\}, $\mathit{SB}_3$ \{13,1,12,9\}, and four non-sensitive bins: $\mathit{NSB}_0$ \{11,10,0,13\}, $\mathit{NSB}_1$ \{2,3,15,1\}, $\mathit{NSB}_2$ \{5,8,6,12\}, $\mathit{NSB}_3$ \{14,7,4,9\}. When a query arrives for a value, say 1, the DB owner searches for the tuples containing values 2,3,15,1 (viz. $\mathit{NSB}_1$) on the non-sensitive data and values in $\mathit{SB}_3$ (viz., 13,1,12,9) on the sensitive data using the cryptographic mechanism integrated into QB. We will show that in the proposed approach, while the adversary learns that the query corresponds to one of the four values in $\mathit{NSB}_1$, since query values in $\mathit{SB}_3$ are encrypted, the adversary does not learn the actual sensitive value or the actual non-sensitive value that is identical to a cleartext sensitive value.

\subsection{The Base Case}
\label{subsec:The Base Case}
QB consists of two steps. First, query bins are created (information about which will reside at the DB owner) using which queries will be rewritten. The second step consists of rewriting the query based on the binning.

Here, QB is explained for the base case, where a sensitive tuple, say $t_s$, is associated with at most a single non-sensitive tuple, say $t_{\mathit{ns}}$, and vice versa (\textit{i}.\textit{e}., $\overset{\mathrm{a}}{=}$ is a 1:1 relationship). Thus, if the value has two tuples, then one of them must be sensitive and the other one must be non-sensitive, but both the tuples cannot be sensitive or non-sensitive. A value can also have only one tuple, either sensitive or non-sensitive. Note that if $t_1, t_2, \ldots, t_l$ are sensitive tuples, with values of an attribute $A$ being $s_1, s_2, \ldots s_n$, $s_i \neq s_j$ if $i\neq j$.

Thus, in the remainder of the section, we will refer to association between encrypted value $E(t_i)[A]$ and a non-sensitive value $\mathit{ns}_j$ simply as an association between values $s_i$ and $\mathit{ns}_j$, where $s_i$ is the cleartext representation of $E(t_i)[A]$ and $\mathit{ns}_j$ is a value in the attribute $A$ of a non-sensitive relation. That is, $s_i \overset{\mathrm{a}}{=} \mathit{ns}_j$ represents $E(t_i)[A] \overset{\mathrm{a}}{=} \mathit{ns}_j$.

The scenario depicted in Example 1 satisfies the base case. The \emph{EId} attribute values corresponding to sensitive tuples include $\langle$\texttt{E101}, \texttt{E259}$, \texttt{E152}, \texttt{E159}\rangle$ and corresponding to non-sensitive tuples are $\langle$\texttt{E199}, \texttt{E259}, \texttt{E254}, \texttt{E152}$\rangle$ for which $\overset{\mathrm{a}}{=}$ is 1:1. We discuss QB under the above assumption, but these assumptions are relaxed in the conference version of this paper (please see \S IV.A and \S IV.B in~\cite{DBLP:conf/icde/Mehrotra0UM19}). Before describing QB, we first define the concept of {\em approximately square factors of a number}.

\medskip
\noindent\textbf{Approximately square factors.} \emph{We say two numbers, say $x$ and $y$, are \emph{approximately square factors of a number}, say $n>0$, if $x \times y =n$, and $x$ and $y$ are equal or close to each other such that the difference between $x$ and $y$ is less than the difference between any two factors, say $x^{\prime}$ and $y^{\prime}$, of $n$ such that  $x^{\prime} \times y^{\prime} = n$.}

\DontPrintSemicolon
\LinesNotNumbered
\begin{algorithm}[!t]
\textbf{Inputs:} $|\mathit{NS}|$: the number of values in the non-sensitive data, $|S|$: the number of values in the sensitive data.

\textbf{Outputs:} $\mathit{SB}$: sensitive bins; $\mathit{NSB}$: non-sensitive bins

\textbf{Variable:} $|\mathit{NSB}|$: non-sensitive values in a non-sensitive bin, $|\mathit{SB}|$: sensitive values in a sensitive bin.


\nl{\bf Function $\mathit{create\_bins(S,NS)}$} \nllabel{ln:function_create_bucket}
\Begin{
\nl Permute all sensitive values \nllabel{ln:permute}

\nl $x, y \leftarrow \mathit{approx\_sq\_factors(|NS|)}$: $x \geq y$ \nllabel{ln:largest_divisors}

\nl $|\mathit{NSB}| \leftarrow x$, $\mathit{NSB} \leftarrow \lceil |\mathit{NS}|/x\rceil$, $\mathit{SB} \leftarrow x$, $|\mathit{SB}| \leftarrow y$ \nllabel{ln:number_of_buckets}

\nl \lFor{$i \in (1,|S|)$}{$\mathit{SB}[i$ modulo $x][\ast]\leftarrow S[i]$\nllabel{ln:sesnitive_allocate}}

\nl \lFor{$(i,j)\in (0,\mathit{SB}-1),(0,|\mathit{SB}|-1)$}{$\mathit{NSB}[j][i]\leftarrow \mathit{allocateNS(\mathit{SB}[i][j])}$ \nllabel{ln:assign_value_to_NS_bucket}}

\nl \lFor{$i\in (0,\mathit{NSB}-1)$}{$\mathit{NSB}[i,\ast]\leftarrow$ fill the bin if empty with the size limit to $x$ \nllabel{ln:assign_remaining_NS}}

\nl \Return $\mathit{SB}$ and $\mathit{NSB}$
}


\nl{\bf Function $\mathit{allocateNS(\mathit{SB}[i][j])}$} \nllabel{ln:function_allocateNS}
\Begin{
find a non-sensitive value associated with the $j^{\mathit{th}}$ sensitive value of the $i^{\mathit{th}}$ sensitive bin
}

\caption{Bin-creation algorithm, the base case.}
\label{alg:bin_creation}
\end{algorithm}
\setlength{\textfloatsep}{0pt}

\medskip
\noindent\textbf{Step 1: Bin-creation}. QB, described in Algorithm~\ref{alg:bin_creation}, finds two approximately square factors of $|\mathit{NS}|$, say $x$ and $y$, where $x\geq y$. QB creates $\mathit{SB}=x$ sensitive bins, where each sensitive bin contains at most $y$ values. Thus, we assume $|S|\geq x$. QB, further, creates $\mathit{NSB}=\lceil|NS|/x\rceil$ non-sensitive bins, where each non-sensitive bin contains at most $|\mathit{NSB}|=x$ values. Note that we are assuming that $|S|\leq|\mathit{NS}|$.\footnote{QB can also handle the case of $|S|>|\mathit{NS}|$ by applying Algorithm~\ref{alg:bin_creation} in a reverse way, \textit{i}.\textit{e}., factorizing $|S|$.}

\smallskip
\noindent\textbf{\emph{Assignment of sensitive values}.} We number the sensitive bins from 0 to $x-1$ and the values therein from 0 to $y-1$. To assign a value to sensitive bins, QB first permutes the set of sensitive values. Such a permutation is kept secret from the adversary by the DB owner.\footnote{We emphasize to first permute sensitive values to prevent the adversary to create bins at her end; \textit{e}.\textit{g}., if the adversary is aware of a fact that employee ids are ordered, then she can also create bins by knowing the number of resultant tuples to a query. However, for simplicity, we do not show permuted sensitive values in any figure.} In order to assign sensitive values to sensitive bins, QB takes the $i^{\mathit{th}}$ sensitive value and assigns it to the $(i$ $\mathit{modulo}$ $x)^{\mathit{th}}$ sensitive bin (see Lines~\ref{ln:permute} and~\ref{ln:sesnitive_allocate} of Algorithm~\ref{alg:bin_creation}).

\smallskip
\noindent\textbf{\emph{Assignment of non-sensitive values}.} We number the non-sensitive bins from 0 to $\lceil|\mathit{NS}|\rceil/x-1$ and values therein from 0 to $x-1$. In order to assign non-sensitive values, QB takes a sensitive bin, say $j$, and its $i^{\mathit{th}}$ sensitive value. Assign the non-sensitive value associated with the $i^{\mathit{th}}$ sensitive value to the $j^{\mathit{th}}$ position of the $i^{\mathit{th}}$ non-sensitive bin. Here, if each value of a sensitive bin has an associated non-sensitive value and $|S| = |\mathit{NS}|$, then QB has assigned all the non-sensitive values to their bins (Line~\ref{ln:assign_value_to_NS_bucket} of Algorithm~\ref{alg:bin_creation}). Note that it may be the case that only a few sensitive values have their associated non-sensitive values and $|S| \leq |\mathit{NS}|$. In this case, we assign the sensitive and their associated non-sensitive values to bins like we did in the previous case. However, we need to assign the non-sensitive values that are not associated with a sensitive value, by filling all the non-sensitive bins to size $x$ (Line~\ref{ln:assign_remaining_NS} of Algorithm~\ref{alg:bin_creation}).


\smallskip
\noindent\emph{Aside}. Note that QB assigned at least as many values in a non-sensitive bin as it assigned to a sensitive bin. QB may form the non-sensitive and sensitive bins in such a way that the number of values in sensitive bins is higher than the non-sensitive bins. We chose sensitive bins to be smaller since the processing time on encrypted data is expected to be higher than cleartext data processing; hence, by searching and retrieving fewer sensitive tuples, we decrease the encrypted data-processing time.

\medskip
\noindent\textbf{Step 2: Bin-retrieval -- answering queries}. Algorithm~\ref{alg:bin_retrieval} presents the pseudocode for the bin-retrieval algorithm. The algorithm, first, checks the existence of a query value in sensitive bins and/or non-sensitive bins (see Lines~\ref{ln:check_each_sbucket} and~\ref{ln:check_nsbuckets} of Algorithm~\ref{alg:bin_retrieval}). If the value exists in a sensitive bin and a non-sensitive bin, the DB owner retrieves the corresponding two bins (see Line~\ref{ln:retrieve_bin}). Note that here the adversarial view is not enough to leak the query value or to find a value that is shared between the two bins. The reason is that the desired query value is encrypted with a set of other encrypted values and, furthermore, the query value is obscured in many requested non-sensitive values, which are in cleartext. Consequently, the adversary is unable to find an intersection of the two bins, which is the exact value.

\DontPrintSemicolon
\LinesNotNumbered \begin{algorithm}[!t]
\textbf{Inputs:} $w$: the query value. $\mathit{SB}$ and $\mathit{NSB}$: Sensitive and non-sensitive bins, created by Algorithm~\ref{alg:bin_creation}.

\textbf{Outputs:} $\mathit{SB}_a$ and $\mathit{NSB}_b$: one sensitive bin and one non-sensitive bin to be retrieved for answering $w$.

\textbf{Variables:} $\mathit{found}\leftarrow$ \textbf{false}

\nl{\bf Function $\mathit{retrieve\_bins(q(w))}$} \nllabel{ln:function_retrieve_bucket}
\Begin{

\nl \For{$(i,j) \in (0,\mathit{SB}-1),(0,|\mathit{SB}|-1)$\nllabel{ln:check_each_sbucket}}{\nl \If{$w=\mathit{SB}_i[j]$}{

\nl \Return $\mathit{SB}_i$ and $\mathit{NSB}_j$; $\mathit{found} \leftarrow$ \textbf{true}; \textbf{break} \nllabel{ln:retrieve_rule_svalue}
}}

\nl \If{$\mathit{found} \neq$ \textbf{\textnormal{\textbf{true}}}\nllabel{ln:check_nsbuckets}}{
\nl \For{$(i,j) \in (0,\mathit{NSB}-1),(0,|\mathit{NSB}|-1)$}{

\nl \If{$w=\mathit{NSB}_i[j]$}{\nl \Return $\mathit{NSB}_i$ and $\mathit{SB}_j$; \textbf{break}\nllabel{ln:retrieve_rule_nsvalue}}}}

\nl Retrieve the desired tuples from the cloud by sending encrypted values of the bin $\mathit{SB}_i$ (or $\mathit{SB}_j$) and cleartext values of the bin $\mathit{NSB}_j$ (or $\mathit{NSB}_i$) to the cloud

\nllabel{ln:retrieve_bin}

}
\caption{Bin-retrieval algorithm.}
\label{alg:bin_retrieval}
\end{algorithm}
\setlength{\textfloatsep}{0pt}

There are the following three other cases to consider:
\begin{enumerate}[noitemsep,leftmargin=0.01in]
\item Some sensitive values of a bin are not associated with any non-sensitive value. For example, in Figure~\ref{fig:qb}, the sensitive values $s_4$, $s_7$, $s_8$, $s_9$, and $s_{10}$ are not associated with any non-sensitive value.

\item A sensitive bin does not hold any value that is associated with any non-sensitive value. For example, the sensitive bin $\mathit{SB}_4$ in Figure~\ref{fig:qb} satisfies this clause.

\item A non-sensitive bin containing no value that is associated with any sensitive value.
\end{enumerate}

In all the three cases, if the DB owner retrieves only either a sensitive or non-sensitive bin containing the value, then it will lead to information leakage similar to Example 2. In order to prevent such leakage, Algorithm~\ref{alg:bin_retrieval} follows two rules stated below (see Lines~\ref{ln:retrieve_rule_svalue} and~\ref{ln:retrieve_rule_nsvalue} of Algorithm~\ref{alg:bin_retrieval}):

\smallskip
\noindent\textbf{Tuple retrieval rule R1.} If the query value $w$ is a sensitive value that is at the $j^{\mathit{th}}$ position of the $i^{\mathit{th}}$ sensitive bin (\textit{i}.\textit{e}., $w=\mathit{SB}_i[j]$), then the DB owner will fetch the $i^{\mathit{th}}$ sensitive and the $j^{\mathit{th}}$ non-sensitive bins (see Line~\ref{ln:retrieve_rule_svalue} of Algorithm~\ref{alg:bin_retrieval}). By Line~\ref{ln:check_each_sbucket} of Algorithm~\ref{alg:bin_retrieval}, the DB owner knows that the value $w$ is either sensitive or non-sensitive.

\smallskip
\noindent\textbf{Tuple retrieval rule R2.} If the query value $w$ is a non-sensitive value that is at the $j^{\mathit{th}}$ position of the $i^{\mathit{th}}$ non-sensitive bin, then the DB owner will fetch the $i^{\mathit{th}}$ non-sensitive and the $j^{\mathit{th}}$ sensitive bins (see Line~\ref{ln:retrieve_rule_nsvalue} of Algorithm~\ref{alg:bin_retrieval}).

\smallskip
Note that if query value $w$ is in both sensitive and non-sensitive bins, then both the rules are applicable, and they retrieve \emph{exactly the same} bins. In addition, if the value $w$ is neither in a sensitive or a non-sensitive bin, then there is no need to retrieve any bin.

\smallskip\noindent\emph{Aside}. After knowing the bins, the DB owner sends all the sensitive values in the encrypted form and the non-sensitive values in cleartext to the cloud. The tuple retrieval based on the encrypted values reveals only the tuple addresses that satisfy the requested values. We can also hide the access-patterns by using PIR, ORAM, or DSSE on each required sensitive value. As mentioned in \S\ref{sec:introduction}, access-pattern-hiding techniques are prone to size and workload-skew attacks. Nonetheless, the use of QB with access-pattern-hiding techniques makes them secure against these attacks.\footnote{QB is designed as a general mechanism that provides partitioned data security when coupled with any cryptographic technique. For special cryptographic techniques that hide access-patterns, it may be possible to design a different mechanism that may provide partitioned data security.}

\smallskip
\noindent
\textbf{Associated bins.} \emph{We say a sensitive bin is associated with a non-sensitive bin, if the two bins are retrieved for answering at least one query.}

Our aim when answering queries for all the sensitive and non-sensitive values using Algorithm~\ref{alg:bin_retrieval} is to associate each sensitive bin with each non-sensitive bin; resulting in the adversary being unable to predict which (if any) is the value shared between two bins.

\begin{figure}[!h]
\centering
\includegraphics[scale=0.6]{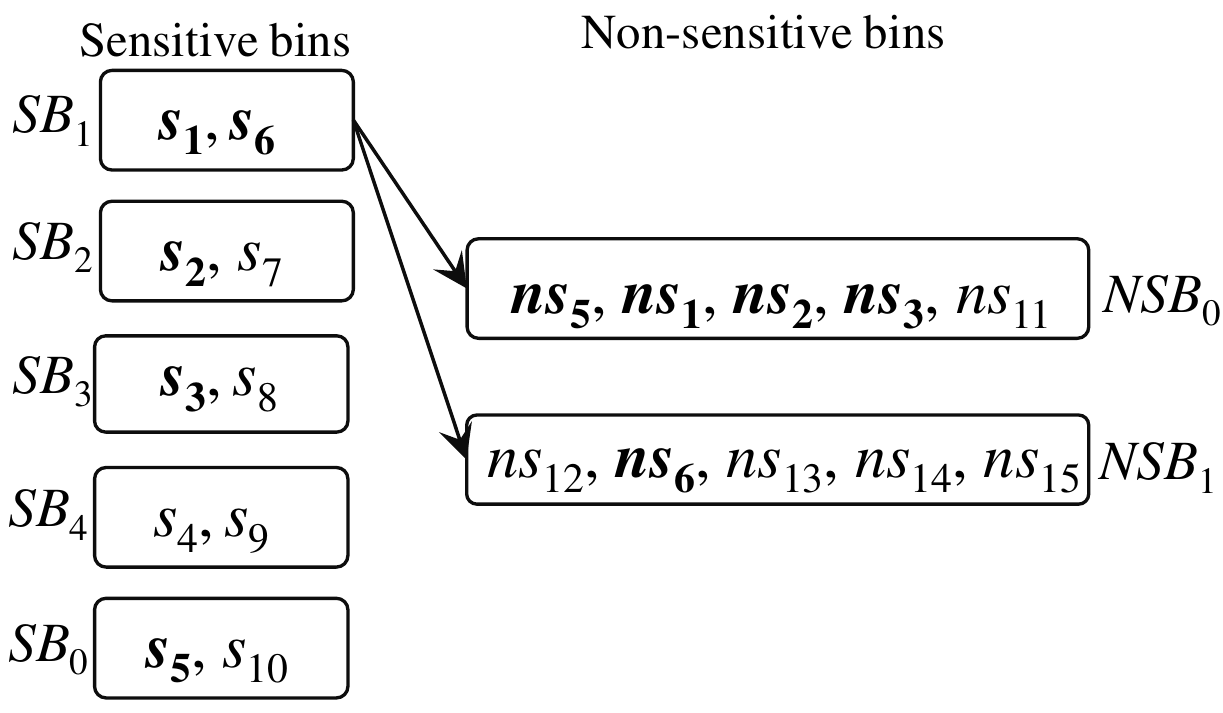}
\caption{QB for 10 sensitive and 10 non-sensitive values.}
\label{fig:qb}
\end{figure}

\smallskip
\noindent\textnormal{\textbf{Example 3: QB example Step 1: Bin Creation.}} We show the bin-creation algorithm for 10 sensitive values and 10 non-sensitive values. We assume that only five sensitive values, say $s_1, s_2, s_3, s_5, s_6$, have their associated non-sensitive values, say $\mathit{ns}_1, \mathit{ns}_2,\mathit{ns}_3,\mathit{ns}_5, \mathit{ns}_{6}$, and the remaining 5 sensitive (say, $s_4, s_7,s_8, \ldots s_{10}$) and 5 non-sensitive values (say, $\mathit{ns}_{11}, \mathit{ns}_{12},\ldots, \mathit{ns}_{15}$) are not associated. For simplicity, we use different indexes for non-associated values.

QB creates 2 non-sensitive bins and 5 sensitive bins, and divides 10 sensitive values over the following 5 sensitive bins: $\mathit{SB}_0$ $\{s_5, s_{10}\}$, $\mathit{SB}_1$ $\{s_1, s_6\}$, $\mathit{SB}_2$ $\{s_2, s_7\}$, $\mathit{SB}_3$ $\{s_3, s_8\}$, $\mathit{SB}_4$ $\{s_4, s_9\}$; see Figure~\ref{fig:qb}. Now, QB distributes non-sensitive values associated with the sensitive values over two non-sensitive bins, resulting in the bin $\mathit{NSB}_0$ $\{\mathit{ns}_5,\mathit{ns}_1,\mathit{ns}_2,\mathit{ns}_3,\ast \}$ and $\mathit{NSB}_1$ $\{\ast,\mathit{ns}_6,\ast,\ast,\ast\}$, where a $\ast$ shows an empty position in the bin. In the sequel, QB needs to fill the non-sensitive bins with the remaining 5 non-sensitive values; hence, $\mathit{ns}_{11}$ is assigned to the last position of the bin $\mathit{NSB}_0$, and the bin $\mathit{NSB}_1$ contains the remaining 4 non-sensitive values such as $\{\mathit{ns}_{12},\mathit{ns}_6,\mathit{ns}_{13},\mathit{ns}_{14},\mathit{ns}_{15}\}$.

\smallskip
\noindent\textbf{Example 3: QB example (continued) Step 2: Bin-retrieval.} Now, we show how to retrieve tuples. If a query is for a sensitive value, say $s_2$ (refer to Figure~\ref{fig:qb}), then the DB owner fetches two bins $\mathit{SB}_2$ and $\mathit{NSB}_0$. If a query is for a non-sensitive value, say $\mathit{ns}_{14}$, then the DB owner fetches two bins $\mathit{NSB}_1$ and $\mathit{SB}_3$. Thus, it is impossible for the adversary to find (by observing the adversarial view) which is an exact query value from the non-sensitive bin and which is the sensitive value associated with one of the non-sensitive values. This fact is also clear from Table~\ref{tab:answer table_qb}, which shows that the adversarial view is not enough to leak information from the joint processing of sensitive and non-sensitive data, unlike Example 2. In Table~\ref{tab:answer table_qb}, $E(s_i)$ shows the encrypted value of $s_i$, and we are showing the adversarial view only for queries for $s_2$, $s_7$, and $\mathit{ns}_{13}$. One may easily create the adversarial view for other queries. In this example, note that the bin $\mathit{SB}_2$ gets associated with both the non-sensitive bins $\mathit{NSB}_0$ and $\mathit{NSB}_1$, due to following Algorithm~\ref{alg:bin_retrieval}.

\begin{table}[h]
  \centering
    \begin{tabular}{|l|l|l|}
    \hline
    \textbf{Exact query value} & \multicolumn{2}{|c|}{\textbf{Returned tuples/Adversarial view}}          \\ \hline
    ~                  & \textbf{Sensitive bin and data}        & \textbf{Non-sensitive bin and data}  \\ \hline\hline

    $s_2$ or $\mathit{ns}_2$ & $\mathit{SB}_2$\textbf{:}$\mathit{E(s_2)}$,$\mathit{E(s_7)}$ & $\mathit{NSB}_0$\textbf{:}$\mathit{ns}_1$,$\mathit{ns}_2$,$\mathit{ns}_3$,$\mathit{ns}_5$,$\mathit{ns}_{11}$ \\\hline

    $s_7$ & $\mathit{SB}_2$\textbf{:}$\mathit{E(s_2)}$,$\mathit{E(s_7)}$ & $\mathit{NSB}_1$\textbf{:}$\mathit{ns}_6$,$\mathit{ns}_{12}$,$\mathit{ns}_{13}$,$\mathit{ns}_{14}$,$\mathit{ns}_{15}$ \\\hline

    $\mathit{ns}_{13}$ & $\mathit{SB}_2$\textbf{:}$\mathit{E(s_2)}$,$\mathit{E(s_7)}$ & $\mathit{NSB}_1$\textbf{:}$\mathit{ns}_6$,$\mathit{ns}_{12}$,$\mathit{ns}_{13}$,$\mathit{ns}_{14}$,$\mathit{ns}_{15}$ \\\hline

    \end{tabular}
    \caption{Queries and returned tuples/adversarial view after retrieving tuples according to Algorithm~\ref{alg:bin_retrieval}.}
    \label{tab:answer table_qb}
\end{table}

\subsection{Algorithm Correctness}
\label{subsec:Algorithm Correctness}
We will prove that QB does not lead to information leakage through the joint processing of sensitive and non-sensitive data. To prove correctness, we first define the concept of \emph{surviving matches}. Informally, we show that QB maintains surviving matches among all sensitive and non-sensitive values, resulting in all sensitive bins being associated with all non-sensitive bins. Thus, an initial condition: a sensitive value is assumed to have an identical value to one of the non-sensitive value is preserved.

\medskip
\noindent\textbf{Surviving matches.} We define surviving matches, which are classified as either \emph{surviving matches of values} or \emph{surviving matches of bins}, as follows:

\smallskip
\noindent\emph{Before query execution}. Observe that before retrieving any tuple, under the assumption that no one except the DB owner can decrypt an encrypted sensitive value, say $E(s_i)$, the adversary cannot learn which non-sensitive value is associated with the value $s_i$. Thus, the adversary will consider that the value $E(s_i)$ is associated with one of the non-sensitive values. Based on this fact, the adversary can create a complete bipartite graph having $|S|$ nodes on one side and $|\mathit{NS}|$ nodes on the other side. The edges in the graph are called \emph{surviving matches of the values}. For example, before executing any query, the adversary can create a bipartite graph for 10 sensitive and 10 non-sensitive values.

\smallskip
\noindent\emph{After query execution}. Recall that the query execution on the datasets creates an adversarial view that guides the adversary to create a (new) bipartite graph containing $\mathit{SB}$ nodes on one side and $\mathit{NSB}$ nodes on the other side. The edges in the new graph (obtained after the query execution) are called \emph{surviving matches of the bins}. For example, after executing queries according to Algorithm~\ref{alg:bin_retrieval}, the adversary can create a bipartite graph having 5 nodes on one side and 2 nodes on the other side, see Figure~\ref{fig:survival_matching2}. Note that since bins contain values, the surviving matches of the bins can lead to the surviving matches of the values. Hence, from Figure~\ref{fig:survival_matching2}, the adversary can also create a bipartite graph for 10 sensitive and 10 non-sensitive values.

We show that a technique for retrieving tuples that drops some surviving matches of the bins leading to drop of the surviving matches of the values is not secure, and hence, results in the information leakage through non-sensitive data.

\begin{figure}[t]
\begin{center}
  \begin{minipage}[t]{.45\linewidth}
  \centering
  \includegraphics[scale=0.7]{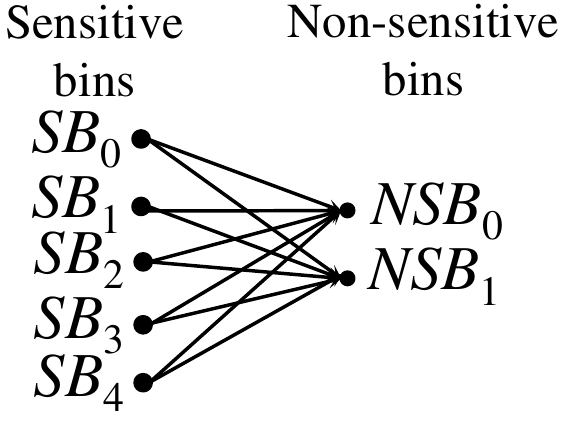}
  \subcaption{Surviving matches after the tuple retrieval following Algorithm~\ref{alg:bin_retrieval}.}
  \label{fig:survival_matching2}
  \end{minipage}
  \begin{minipage}[t]{.45\linewidth}
  \centering
  \includegraphics[scale=0.7]{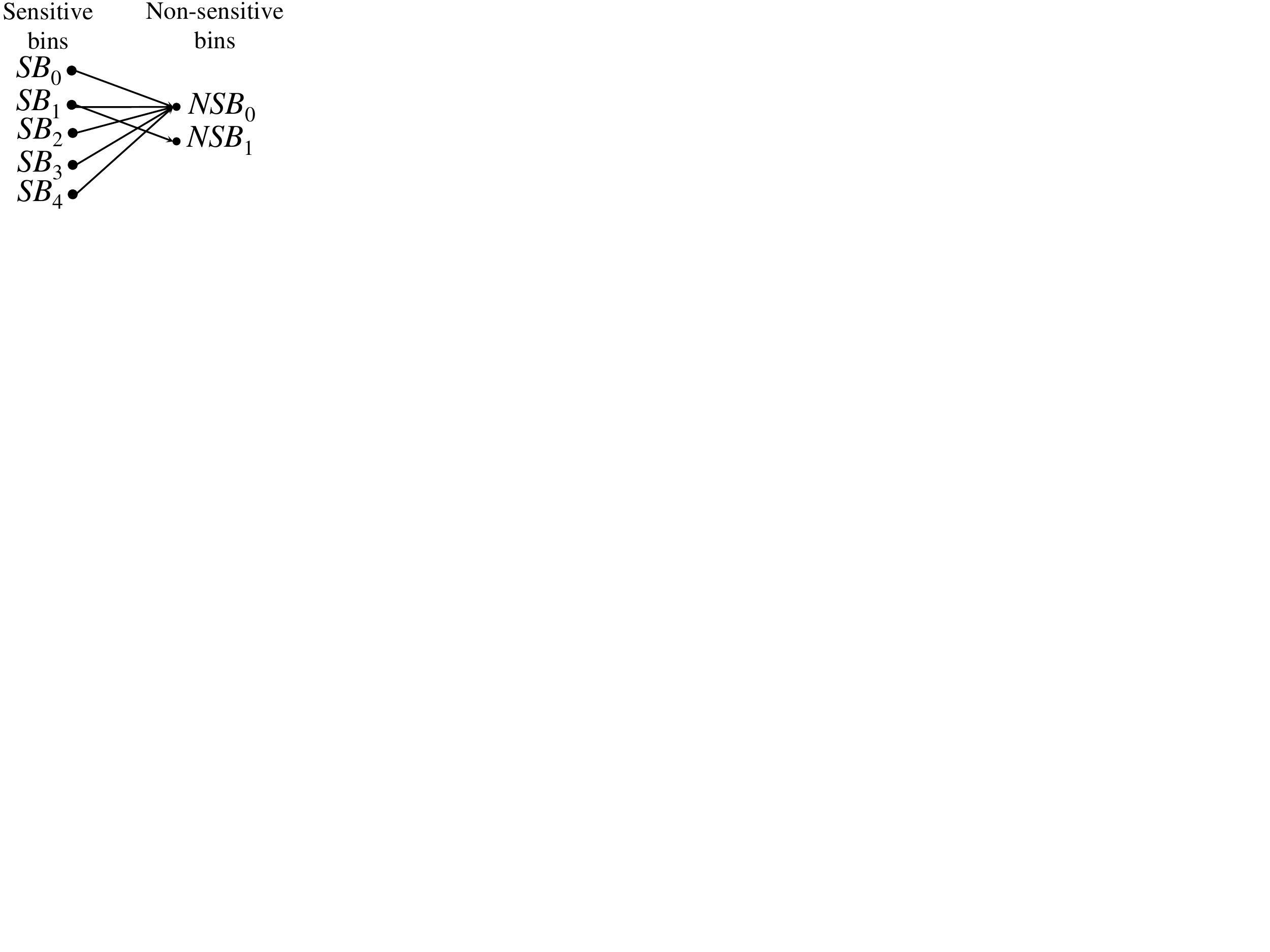}
  \subcaption{Surviving matches without following Algorithm~\ref{alg:bin_retrieval} for $\mathit{ns}_{12}$, $\mathit{ns}_{13}$, $\mathit{ns}_{14}$, $\mathit{ns}_{15}$; also see Table~\ref{tab:answer table_qb}.}
  \label{fig:survival_matching3}
  \end{minipage}
\end{center}
\caption{An example to show security of QB using surviving matches for 10 sensitive and 10 non-sensitive values.}
\label{fig:survival_mathcings}
\FFF\FFF
\end{figure}

\smallskip
\noindent\textbf{Example 4: Dropping surviving matches.} In Figure~\ref{fig:qb}, for answering queries for associated values $s_1$, $s_2$, $s_3$, $s_5$, $s_6$, $\mathit{ns}_1$, $\mathit{ns}_2$, $\mathit{ns}_3$, $\mathit{ns}_5$, or $\mathit{ns}_6$, the DB owner must follow Line~\ref{ln:retrieve_rule_svalue} or~\ref{ln:retrieve_rule_nsvalue} of Algorithm~\ref{alg:bin_retrieval} for retrieving the two bins holding corresponding sensitive and non-sensitive data; otherwise, the DB owner cannot retrieve two bins that share a common value. Now, retrieved tuples for these values create an adversarial view as shown in the first six lines except the fourth line  of Table~\ref{tab:answer table_qb}. However, for answering values $s_4$, $s_7$, $s_8$, $s_9$, $s_{10}$, $\mathit{ns}_6$, $\mathit{ns}_{12}$, $\mathit{ns}_{13}$, $\mathit{ns}_{14}$, or $\mathit{ns}_{15}$ (recall that these values are not associated), if the DB owner does not follow Algorithm~\ref{alg:bin_retrieval} and retrieves the bin containing the desired value with any randomly selected bin of the other side, then it could result in the following adversarial view; see Table~\ref{tab:wrong_view}. We show the case when $\mathit{NSB}_1$ is only associated with bin $\mathit{SB}_1$, and bins $\mathit{SB}_2$ is only associated with bin $\mathit{NSB}_0$, since
Algorithm~\ref{alg:bin_retrieval} is not followed.

\begin{table}[t]
  \centering
    \begin{tabular}{|l|l|l|}
    \hline
    \textbf{Exact query value} & \multicolumn{2}{|c|}{\textbf{Returned tuples/Adversarial view}}           \\ \hline
    ~                  & \textbf{Sensitive bin and data}        & \textbf{Non-sensitive bin and data}  \\ \hline\hline
%
    $s_2$ or $\mathit{ns}_2$ & $\mathit{SB}_2$\textbf{:}$\mathit{E(s_2)}$,$\mathit{E(s_7)}$ & $\mathit{NSB}_0$\textbf{:}$\mathit{ns}_1$,$\mathit{ns}_2$,$\mathit{ns}_3$,$\mathit{ns}_5$,$\mathit{ns}_{11}$ \\\hline
%







    $s_6$ or $\mathit{ns}_6$ & $\mathit{SB}_1$\textbf{:}$\mathit{E(s_1)}$,$\mathit{E(s_6)}$ & $\mathit{NSB}_1$\textbf{:}$\mathit{ns}_6$,$\mathit{ns}_{12}$,$\mathit{ns}_{13}$,$\mathit{ns}_{14}$,$\mathit{ns}_{15}$ \\\hline

    $s_7$ & $\mathit{SB}_2$\textbf{:}$E(s_2)$,$E(s_7)$ &
    $\mathit{NSB}_0$\textbf{:}$\mathit{ns}_1$,$\mathit{ns}_2$,$\mathit{ns}_3$,$\mathit{ns}_5$,$\mathit{ns}_{11}$ \\\hline

    $\mathit{ns}_{12}$ & $\mathit{SB}_1$\textbf{:}$E(s_1)$,$E(s_6)$ & $\mathit{NSB}_1$\textbf{:}$\mathit{ns}_6$,$\mathit{ns}_{12}$,$\mathit{ns}_{13}$,$\mathit{ns}_{14}$,$\mathit{ns}_{15}$ \\\hline

    $\mathit{ns}_{13}$ & $\mathit{SB}_1$\textbf{:}$E(s_1)$,$E(s_6)$ & $\mathit{NSB}_1$\textbf{:}$\mathit{ns}_6$,$\mathit{ns}_{12}$,$\mathit{ns}_{13}$,$\mathit{ns}_{14}$,$\mathit{ns}_{15}$ \\\hline

    $\mathit{ns}_{14}$ & $\mathit{SB}_1$\textbf{:}$E(s_1)$,$E(s_6)$ & $\mathit{NSB}_1$\textbf{:}$\mathit{ns}_6$,$\mathit{ns}_{12}$,$\mathit{ns}_{13}$,$\mathit{ns}_{14}$,$\mathit{ns}_{15}$ \\\hline

    $\mathit{ns}_{15}$ & $\mathit{SB}_1$\textbf{:}$E(s_1)$,$E(s_6)$ & $\mathit{NSB}_1$\textbf{:}$\mathit{ns}_6$,$\mathit{ns}_{12}$,$\mathit{ns}_{13}$,$\mathit{ns}_{14}$,$\mathit{ns}_{15}$ \\\hline

    \end{tabular}
    \caption{Queries and returned tuples/adversarial view without following Algorithm~\ref{alg:bin_retrieval}.}
    \label{tab:wrong_view}
    \FFF\FFF
\end{table}

Having such an adversarial view (Table~\ref{tab:wrong_view}), the adversary can learn two facts that
\begin{enumerate}[noitemsep,leftmargin=0.01in]
  \item Encrypted sensitive tuples of the bin $\mathit{SB}_2$ have associated non-sensitive tuples only in the bin $\mathit{NSB}_0$, not in $\mathit{NSB}_1$ (Figure~\ref{fig:survival_matching3}).
  \item Non-sensitive tuples of the bin $\mathit{NSB}_1$ have their associated sensitive tuples only in the bin $\mathit{SB}_1$ (see Figure~\ref{fig:survival_matching3}).
\end{enumerate}
Based on this adversarial view (Table~\ref{tab:wrong_view}), the bipartite graph drops some surviving matches of the bins (see Figure~\ref{fig:survival_matching3}). (That fact leads to the dropping of the surviving matches of the values, specifically, surviving matches between sensitive values $s_3$, $s_4$, $s_5$, $s_8$, $s_9$, $s_{10}$ and non-sensitive value $\mathit{ns}_6$, $\mathit{ns}_{12}$, $\mathit{ns}_{13}$, $\mathit{ns}_{14}$, $\mathit{ns}_{15}$.) Hence, a random retrieval of bins is not a secure technique to prevent information leakage through non-sensitive data accessing.

In contrast, if the DB owner uses Line~\ref{ln:retrieve_rule_svalue} or~\ref{ln:retrieve_rule_nsvalue} of Algorithm~\ref{alg:bin_retrieval} for retrieving values that are not associated, the above-mentioned facts (\textit{i}) and (\textit{ii}) no longer hold. Figure~\ref{fig:survival_matching2} shows the case when each sensitive bin is associated with each non-sensitive bin, if Algorithm~\ref{alg:bin_retrieval} is followed. Thus, we can see that all the surviving matches of the bins and values are preserved after answering queries. Therefore, for the example of 10 sensitive and 10 non-sensitive values, QB (Algorithms~\ref{alg:bin_creation} and~\ref{alg:bin_retrieval}) is secure, and under the given assumptions (\S\ref{subsec:Security Definition and Correctness}), the adversary cannot find an exact association between a sensitive and a non-sensitive value.

{
\smallskip
\noindent
\subsection*{Security Proof}
Now, we prove that QB is secure and satisfies the definition of partitioned data security (Theorem~\ref{th:Preserve the perfect data security}) by first proving that all the sensitive bins are associated with all the non-sensitive bins (Theorem~\ref{th:algorithm correctness}), which is intuitively clear by Example 4. Recall that the only way a surviving match could be removed is if there is no sensitive value in a sensitive bin, say $\mathit{SB}_j$ that does not have an associated non-sensitive value. In this case for answering a value belonging to $\mathit{SB}_j$, we retrieve either only the bin $\mathit{SB}_j$ or the bin $\mathit{SB}_j$ with any randomly selected non-sensitive bin. Note that the adversary cannot learn anything from the encrypted data, since the keys are only known to the DB owner.

\begin{theorem}
\label{th:algorithm correctness}
Let $|S|$ and $|\mathit{NS}|$ be the number of sensitive and non-sensitive values, respectively. By following Algorithm~\ref{alg:bin_creation}, $|S|$ and $|\mathit{NS}|$ values are distributed over $\mathit{SB}$ sensitive and $\mathit{NSB}$ non-sensitive bins, respectively. Answering a set of queries using QB (Algorithm~\ref{alg:bin_retrieval}) will not remove any surviving matches of the bins and that leads to preserve all the surviving matches of the values.
\end{theorem}
\begin{proof}
We show that QB will not remove any surviving matches of the bins by showing that a sensitive bin, say $\mathit{SB}_j$, must be associated with all the non-sensitive bins. A similar argument can be proved for any non-sensitive bin. Let $y$ be the number of sensitive values in the bin $\mathit{SB}_j$, and let $p\geq y$, ($p=\mathit{NSB}$) be the number non-sensitive bins. We will prove the following three arguments:
\begin{enumerate}[noitemsep,leftmargin=0.01in]
  \item If a sensitive value, say $s_i\in \mathit{SB}_j$, is associated with a non-sensitive value (\textit{i}.\textit{e}., $\exists \mathit{ns}_z\in R_{\mathit{ns}}: \mathit{ns}_z \overset{\mathrm{a}}{=} s_i$), then two bins, $\mathit{SB}_j$, and one non-sensitive bin, holding the value $\mathit{ns}_z$, are retrieved.

  \item If a sensitive value, say $s_i\in \mathit{SB}_j$, is not associated with any non-sensitive value (\textit{i}.\textit{e}., $\forall \mathit{ns}_j \in R_{\mathit{ns}}: s_i \overset{\mathrm{a}}{\neq} \mathit{ns}_j$), then the bin $\mathit{SB}_j$ and one of the non-sensitive bins are retrieved. Following that, if all the sensitive values of the bins $\mathit{SB}_j$ are not associated with any non-sensitive value (\textit{i}.\textit{e}., $\forall \mathit{ns}_j \in R_{\mathit{ns}}, \forall s_i\in \mathit{SB}_j : s_i \overset{\mathrm{a}}{\neq} \mathit{ns}_j$), then the bin $\mathit{SB}_j$ and $y$ different non-sensitive bins are retrieved.

      By proving the first and second arguments, we will show that if there are \emph{only} $y$ non-sensitive bins, then a sensitive bin must be associated with all the $y$ non-sensitive bins. The following third argument will consider more than $y$ non-sensitive bins.

  \item If there are more than $y$ non-sensitive bins (say, $\mathit{NSB}_y, \mathit{NSB}_{y+1}, \ldots, \mathit{NSB}_p$) having $x$ values that are not associated with any sensitive value (\textit{i}.\textit{e}., $\forall \mathit{ns}_j \in \mathit{NSB}_y \vee \mathit{NSB}_{y+1} \vee \ldots \vee \mathit{NSB}_p, \mathit{ns}_j \overset{\mathrm{a}}{\neq} s_i, i=1, 2, \ldots, |S|$), then each of these non-sensitive bins must be associated with the bin $\mathit{SB}_j$.
\end{enumerate}
By satisfying the above three arguments, we prove that, thus, the bin $\mathit{SB}_j$ is associated with all non-sensitive bins, and hence, all surviving matches of the bins and, eventually, values are preserved.

\noindent\emph{First case}. The value $s_i$ is allocated to $(i$ $\mathit{modulo}$ $x)^{\mathit{th}}$ sensitive bin at an index, say $z$, where $z = 0, 1, \ldots y-1$, and its associated non-sensitive value is allocated to the $(i$ $\mathit{modulo}$ $x)^{\mathit{th}}$ position of the $z^{\mathit{th}}$ non-sensitive bin. When answering a query for $s_i$ according to the rule R1, the bin $\mathit{SB}_j$ with the bin $\mathit{NSB}_z$ are retrieved. Consequently, the desired tuples containing $s_i$ and its associated non-sensitive value are retrieved, and that are correct answers to the query.

\noindent\emph{Second case}. When answering a query for the value $s_i=\mathit{SB}_j[u]$ ($u \in 0, 1, y-1$) that does not have any associated non-sensitive value, by following the rule R1, the bin $\mathit{SB}_j$ with one of the non-sensitive bin $\mathit{NSB}_u$ are retrieved. Moreover, answering queries for all the $y$ values (0, 1, $y-1$) of the bin $\mathit{SB}_j$, by following rule R1, requires us to retrieve the $\mathit{SB}_j$ with all the $y-1$ (0, 1, $y-1$) non-sensitive bins.

\noindent\emph{Third case}. Since the non-sensitive bin, say $\mathit{NSB}_z$, where $z=y, y+1, \ldots, p$, must hold a value at the $j^{\mathit{th}}$ position, by following the rule R2, the bin $\mathit{NSB}_z$ and the sensitive bin $\mathit{SB}_j$ are fetched for answering a query for $\mathit{ns}_j$.

Therefore, the bin $\mathit{SB}_j$ is associated with all the non-sensitive bins, and hence, all the surviving matches between the values of the bin $\mathit{SB}_j$ and all the non-sensitive bins are also maintained.
\end{proof}

Since we proved all sensitive bins are associated with all the non-sensitive bins, based on this fact, we will show that the first condition of partitioned data security holds to be true for any query. Here, we do not show the second equation of partitioned data security definition (\textit{i}.\textit{e}., $Pr_{\mathit{adv}}[s_i \overset{\mathrm{r}}{\sim} s_j|X] = Pr_{\mathit{adv}}[s_i \overset{\mathrm{r}}{\sim} s_j |X, q(w)(R_s,R_\mathit{ns})[A]]$); recall that here in the base case, we assumed that a value has only a single sensitive tuple; hence, the condition holds true.
\begin{theorem}\textnormal{\textbf{(Preserve partitioned data security)}}
\label{th:Preserve the perfect data security}
Let $R$ be a relation containing sensitive and non-sensitive tuples. Let $R_s$ and $R_{\mathit{ns}}$ be the sensitive and non-sensitive relations, respectively. Let $q(w)(R_s,R_{\mathit{ns}})[A]$ be a query, $q$, for a value $w$ in the attribute $A$ of the $R_s$ and $R_{\mathit{ns}}$ relations. Let $X$ be the auxiliary information about the sensitive data, and $\mathit{Pr_{Adv}}$ be the probability of the adversary knowing any information. Let $e_i$ be the $i^{\mathit{th}}$ sensitive tuple value in the attribute $A$ of the relation $R_s$ and $\mathit{ns}_j$ is the $j^{\mathit{th}}$ non-sensitive value in the attribute $A$ of the relation $R_{\mathit{ns}}$. The execution of a set of queries on the attribute $A$ on the relations using QB leads to the following equation to be true: $$Pr_{\mathit{adv}}[e_i \overset{\mathrm{a}}{=} \mathit{ns}_j|X] = Pr_{\mathit{adv}}[e_i \overset{\mathrm{a}}{=} \mathit{ns}_j|X, \mathit{AV}]$$ where $i\in 1, 2,\ldots, |S|$ and $j\in 1, 2,\ldots, |\mathit{NS}|$.
\end{theorem}
\noindent\emph{Proof sketch}. We provide an example of four values to show the correctness of the above theorem. Let $v_1$, $v_2$, $v_3$, and $v_4$ be values containing only one sensitive and one non-sensitive tuple. Let $E_1$, $E_2$, $E_3$, and $E_4$ be encrypted representations of these values in an arbitrary order, \textit{i}.\textit{e}., it is not mandatory that $E_1$ is the encrypted representation of $v_1$. In this example, the cloud stores an encrypted relation, say $R_s$, containing four encrypted tuples with encrypted representations $E_1$, $E_2$, $E_3$, $E_4$ and a cleartext relation, say $R_{\mathit{ns}}$, containing four cleartext tuples with values $v_1$, $v_2$, $v_3$, $v_4$. The objective of the adversary is to deduce a cleartext value corresponding to an encrypted value. Note that before executing a query, the probability of an encrypted value, say $E_i$, to have the cleartext value, say $v_i$, $1\leq i\leq 4$ is 1/4, which QB maintains at the end of a query.

Assume that the user wishes to retrieve the tuple containing $v_1$. By following QB, the user asks a query, say $q(E_1,E_3)(R_s)$, on the encrypted relation $R_s$ for $E_1$, $E_3$, and a query, say $q(v_1,v_2)(R_{\mathit{ns}})$, on the cleartext relation $R_{\mathit{ns}}$ for $v_1,v_2$. After executing the queries, the adversary holds an adversarial view given in Table~\ref{tab:correct_view1}.

\begin{table}[h]
  \centering
    \begin{tabular}{|l|l|l|}
    \hline
    \textbf{Exact query value (hidden from adversary)} & \multicolumn{2}{|c|}{\textbf{Returned tuples/Adversarial view}}           \\ \hline
    ~                  & \textbf{Sensitive data }       & \textbf{Non-sensitive data}  \\ \hline\hline
    $v_1$ & $E_1$,$E_3$ & $v_1$,$v_2$ \\\hline

    \end{tabular}
    \caption{Queries and returned tuples/adversarial view after executing a query for $v_1$, by following Algorithm~\ref{alg:bin_retrieval}.}
    \label{tab:correct_view1}
\end{table}

In this example, we show that the probability of finding the cleartext value of an encrypted representation, say $E_i$, $1\leq i \leq 4$, remains identical before and after a query. In order to show that when a query comes for $2\times \sqrt{n}$ values by following QB, where $n$ is the number of values in the non-sensitive relation, $\sqrt{n}$ values are asked for the sensitive relation and $\sqrt{n}$ values are asked for the non-sensitive relation, we need to figure out:

\begin{enumerate}[noitemsep,leftmargin=0.01in]
  \item All possible allocations of the non-sensitive $\sqrt{n}$ values, say $v_1, v_2, \ldots, v_{\sqrt{n}}$, to $\sqrt{n}$ encrypted sensitive values, say $E_1, E_2, \ldots, E_{\sqrt{n}}$. Here, we use the term \emph{allocation} to show the fact that the encrypted representation of $E_i$ has the cleartext value $v_i$.

      In our example of four values, we find allocations of four non-sensitive values $v_1$, $v_2$, $v_3$, $v_4$ to encrypted representation $E_1$, $E_2$, $E_3$, $E_4$.

  \item All possible allocations of $\sqrt{n}$ non-sensitive values, except one non-sensitive value, say $v_i$, that is allocated to an encrypted sensitive value, say $E_i$, to the remaining encrypted sensitive values.

      In the case of four values and above-mentioned queries, we find allocations of the non-sensitive values $v_2$, $v_3$, $v_4$ to the encrypted sensitive values $E_2$, $E_3$, $E_4$ while assuming that the encrypted representation of $v_1$ is $E_1$.
\end{enumerate}

The ratio of the above two provides the probability of finding a cleartext value corresponding to its encrypted value after the query execution.

When the query arrives for $\langle E_1,E_3,v_1,v_2\rangle$, the adversary gets the fact that the cleartext representation of $E_1$ and $E_3$ cannot be $v_1$ and $v_2$ or $v_3$ and $v_4$. If this will happen, then there is no way to associate a sensitive bin with each non-sensitive bin. Now, if the adversary considers the cleartext representation of $E_1$ is $v_1$, then the adversary has the following four possible allocations of the values $v_1$, $v_2$, $v_3$, $v_4$ to $E_1$, $E_2$, $E_3$, $E_4$:
\begin{center}
$\langle v_1,v_2,v_3,v_4\rangle$,
$\langle v_1,v_2,v_4,v_3\rangle$,

$\langle v_1,v_3,v_4,v_2\rangle$,
$\langle v_1,v_4,v_3,v_2\rangle$.
\end{center}

However, the allocations $\langle v_1,v_3,v_2,v_4\rangle$ and $\langle v_1,v_4,v_2,v_3\rangle$ to $E_1$, $E_2$, $E_3$, and $E_4$ cannot exist. Since the adversary is not aware of the exact cleartext value of $E_1$, the adversary also considers the cleartext representation of $E_1$ is $v_2$. This results in four more possible allocations of the values to $E_1$, $E_2$, $E_3$, and $E_4$, as follows:
\begin{center}
$\langle v_2,v_1,v_3,v_4\rangle$,
$\langle v_2,v_1,v_4,v_3\rangle$,

$\langle v_2,v_3,v_4,v_1\rangle$,
$\langle v_2,v_4,v_3,v_1\rangle$.
\end{center}

However, $\langle v_2,v_3,v_1,v_4\rangle$ and $\langle v_2,v_4,v_1,v_3\rangle$ cannot exist. Similarly, assuming the cleartext representation of $E_1$  is $v_3$ or $v_4$,
we get the following 8 more possible allocations of the values to $E_1$, $E_2$, $E_3$, and $E_4$:
\begin{center}
$\langle v_3,v_1,v_2,v_4\rangle$,
$\langle v_3,v_2,v_1,v_4\rangle$,

$\langle v_3,v_4,v_1,v_2\rangle$,
$\langle v_3,v_4,v_2,v_1\rangle$,

$\langle v_4,v_1,v_2,v_3\rangle$,
$\langle v_4,v_2,v_1,v_3\rangle$,

$\langle v_4,v_3,v_1,v_2\rangle$,
$\langle v_4,v_3,v_2,v_1\rangle$.
\end{center}

Here, the following four allocations of the values to encrypted representation cannot exist:
\begin{center}
$\langle v_3,v_1,v_4,v_2\rangle$, $\langle v_3,v_2,v_4,v_1\rangle$,

$\langle v_4,v_1,v_3,v_2\rangle$,  $\langle v_4,v_2,v_3,v_1\rangle$.
\end{center}

Thus, the retrieval of the four tuples containing one of the following: $\langle E_1,E_3,v_1,v_2\rangle$, results in 16 possible allocations of the values $v_1$, $v_2$, $v_3$, and $v_4$ to $E_1$, $E_2$, $E_3$, and $E_4$, of which only four possible allocations have $v_1$ as the cleartext representation of $E_1$. This results in the probability of finding $E_1=v_1$ is 1/4. A similar argument also holds for other encrypted values. Hence, an initial probability of associating a sensitive value with a non-sensitive value remains identical after executing a query.

Thus, we can conclude the following:
\begin{enumerate}[noitemsep,leftmargin=0.01in]
  \item All possible allocations of $\sqrt{n}$ non-sensitive values, except one non-sensitive value, say $v_1$, that we allocate to an encrypted sensitive value, say $E_1$, to the remaining encrypted sensitive values is $(n-1)! - x$, where $n$ is the number of values in the non-sensitive relation and $x$ is the number of allocations of values $v_2, v_3, \ldots, v_{\sqrt{n}}$ to $E_2, E_3, \ldots, E_{\sqrt{n}}$ that cannot exist.

  \item All possible allocations of the non-sensitive $\sqrt{n}$ values, say $v_1, v_2, \ldots, v_{\sqrt{n}}$, to $\sqrt{n}$ encrypted sensitive values, say $E_1, E_2, \ldots, E_{\sqrt{n}}$, is $n\times ((n-1)! - x)$. This is true because  we cannot allocate any combination of the values asked in the query to any encrypted representations that are asked by the query.
\end{enumerate}

Thus, the retrieval of $2\times \sqrt{n}$ values results in $n\times ((n-1)! - x)$ possible allocations of $\sqrt{n}$ non-sensitive values to $\sqrt{n}$ encrypted sensitive values, while $(n-1)! - x$ allocations exist when a queried non-sensitive value is assumed to be the cleartext of a queried encrypted representation. Therefore, the probability of finding the exact allocation of the non-sensitive values to encrypted sensitive value while considering a non-sensitive value is the cleartext of an encrypted value is $\frac{(n-1)! - x}{n\times ((n-1)! - x)}= \frac{1}{n}$.
}

\smallskip\noindent\textbf{Note: Handling adaptive adversaries.} The above-presented approach can handle an honest-but-curious adversary, who cannot execute any query, and the case when only the DB owner executes the queries on the databases. Now, we show how to handle an adaptive adversary that can execute queries on the database based on the result of previously selected queries. Note that an adaptive adversary can use any bin structure to break QB. She may ask some queries on the non-sensitive data and some queries on the sensitive data. Her objective is to find a value that is common in sensitive and non-sensitive datasets.

We explain with the help of an example that shows how an adaptive adversary breaks QB. Consider four sensitive tuples having sensitive value, say $s_1,$ $s_2,$ $\ldots,$ $s_4$, and four non-sensitive tuples having non-sensitive values, say $\mathit{ns}_1,$ $\mathit{ns}_2,$ $\ldots,$ $\mathit{ns}_4$. Suppose that $s_i$ is associated with $\mathit{ns}_i$, and all sensitive tuples are encrypted. A correct bin structure (not considering permuted sensitive values) will be as follows: $\mathit{SB}_1$: $\{s_1,s_3\}$, $\mathit{SB}_0$: $\{s_2,s_4\}$, $\mathit{NSB}_1$: $\{\mathit{ns}_1,\mathit{ns}_2\}$, and $\mathit{NSB}_0$: $\{\mathit{ns}_3,\mathit{ns}_4\}$.

Now, first see how an adaptive adversary can break QB, with the help of two queries: Consider the first query for $\mathit{ns}_1$. The adversary can ask the query for $\mathit{ns}_1$, $\mathit{ns}_2$, $\mathit{s}_1$, and $\mathit{s}_2$. The adversary will learn that the first and second encrypted tuples are returned. However, she cannot know which of the tuple has an encrypted representation of $s_1$.

Another query is for $\mathit{ns}_3$, and  she asks for $\mathit{ns}_1$, $\mathit{ns}_3$, $\mathit{s}_1$, and $\mathit{s}_3$. The adversary will learn that the first and third encrypted tuples are returned. However, now, she will know that the first encrypted tuple has the encrypted representation is $s_1$, because it was retrieved in the first query for $\mathit{ns}_1$ as well as in the second query.\footnote{Of course, if the encrypted relation does not have any tuple having $s_1$, then the adversary can learn that $\mathit{ns}_1$ is not associated with any tuple. However, this can be prevented trivially by outsourcing fake tuples having $s_1$.} Thus, by observing access-patterns, the adversary can know which two tuples are associated.

To protect this attack, we need to use a cryptographic technique, \textit{e}.\textit{g}., ORAM or secret-sharing that hides access-patterns at the sensitive data. When using access-patterns-hiding cryptographic techniques, the adversary will learn only the fact that two tuples are returned in response to any query. But it will not lead to any inference attacks. It is important to recall that access-patterns-hiding cryptographic techniques are prone to output size attacks. Thus, when mixing these techniques with QB makes them secure against output-size attacks. Note that we cannot use SGX-based solutions at the encrypted data when dealing with an adaptive adversary, because the adversary can observe access-patterns due to cache-lines and branch shadowing~\cite{DBLP:conf/eurosec/GotzfriedESM17,DBLP:conf/ccs/WangCPZWBTG17}.

{
\smallskip\noindent\textbf{Note: Security offered by existing cryptographic techniques vs QB.} Papers such as~\cite{DBLP:conf/ndss/IslamKK12,DBLP:conf/ccs/NaveedKW15,DBLP:conf/ccs/CashGPR15,DBLP:conf/ccs/KellarisKNO16,DBLP:conf/sp/GrubbsSB0R17} have illustrated that formal security guarantees (\textit{e}.\textit{g}., as often shown in papers such as property preserving encryption~\cite{DBLP:conf/sigmod/AgrawalKSX04,DBLP:conf/crypto/BellareBO07,DBLP:journals/cacm/PopaRZB12} and symmetric searchable encryption~\cite{DBLP:journals/jcs/CurtmolaGKO11}) does not prevent leakage through inferences. For instance, Naveed et al.~\cite{DBLP:conf/ccs/NaveedKW15} showed that a cryptographically secured database that is also using an order-preserving cryptographic technique (\textit{e}.\textit{g}., order-preserving encryption (OPE)) may reveal the entire data when mixed with publicly known databases. Note that in our setting, the proposed technique, where the non-sensitive data resides in cleartext, would offer almost no security without query binning. In particular, if the cryptographic technique used to store sensitive data reveals access-patterns, then the adversary will learn about which ciphertext corresponds to which keyword simply by observing the queries on cleartext. Such inferences are prevented by query binning. Also, note that unlike the security properties of searchable encryption techniques (\textit{e}.\textit{g}., OPE, deterministic encryption, and symmetric searchable encryption), which formalize security as indistinguishability  from chosen keyword attack (IND-CKA1)~\cite{DBLP:journals/jcs/CurtmolaGKO11} other than what can be inferred from the permitted leakages,
our scheme does not lead to any leakage due to the joint processing of sensitive and non-sensitive datasets. Thus, QB is safe from inference attacks, and using QB in conjunction with any cryptographic technique does not lead to any additional leakages.
}

{
\subsection{A Simple Extension of the Base Case}
\label{subsec:A Simple Extension of Query Bucketization}
Algorithm~\ref{alg:bin_creation} creates bins when the number of non-sensitive data values\footnote{Recall that we considered the case of $|S|\leq |\mathit{NS}|$.} is not a prime number, by finding the two approximately square factors. However, Algorithm~\ref{alg:bin_creation} may exhibit a relatively higher \emph{cost} (\textit{i}.\textit{e}., the number of the retrieved tuple) when the sum of the approximately square factors is high.

\parskip 0pt
\setlength{\parindent}{15pt}

For example, if there are 41 sensitive data values and 82 non-sensitive data values, then Algorithm~\ref{alg:bin_creation} creates 2 non-sensitive bins having 41 values in each and 41 sensitive bins having exactly one value in each (Line~\ref{ln:number_of_buckets} of Algorithm~\ref{alg:bin_creation}). Consequently, answering a query results in retrieval of 42 tuples. (We may also create two sensitive bins and 41 non-sensitive bins containing exactly two non-sensitive values in each, resulting in retrieval of 23 tuples.) However, the cost can be further reduced by a significant amount, which is explained below.

\smallskip
\noindent\textnormal{\textbf{Example 5: (An example of QB extension --- Algorithm~\ref{alg:bin_extension}).}} Consider again the example of 41 sensitive and 82 non-sensitive values. In this case, 81 is the closest square number to 82. Here, Algorithm~\ref{alg:bin_extension}, described next, creates 9 non-sensitive bins and 9 sensitive bins. By Lines~\ref{ln:sesnitive_allocate} and~\ref{ln:assign_value_to_NS_bucket} of Algorithm~\ref{alg:bin_creation}, sensitive values and associated non-sensitive values are allocated, resulting in that a sensitive bin holds at most 5 values and a non-sensitive bin holds at most 10 values. Thus, at most 15 tuples are retrieved to answer a query.

\DontPrintSemicolon
\LinesNotNumbered
\begin{algorithm}[!t]
{
\textbf{Inputs:} $|\mathit{NS}|$, $|S|$.

\textbf{Outputs:} $\mathit{SB}$, $\mathit{NSB}$

\nl{\bf Function $\mathit{bin\_extension(S,NS)}$} \nllabel{ln:function_extension}
\Begin{

\nl Permute all sensitive values \nllabel{ln:permute}

\nl $x, y \leftarrow \mathit{approx\_sq\_factors(|NS|)}$: $x \geq y$; $\mathit{cost_d} \leftarrow x + y$ \nllabel{ln:largest_divisors_extension}


\nl $z \leftarrow \mathit{closest\_SquareNum}(|\mathit{NS}|)$, $\mathit{cost_{sn}} \leftarrow 2(z/ \sqrt{z})$ \nllabel{ln:find_sn}


\nl \If{$(\mathit{cost_{sn}}+ \lceil (|\mathit{NS}|-z)/\sqrt{z}\rceil < \mathit{cost_d})$}{\nllabel{ln:cost_check}

\nl Execute Algorithm~\ref{alg:bin_creation}$(S,z)$ and add $(\mathit{NS}-z)/\sqrt{z}$ number of the \emph{remaining} non-sensitive values in each non-sensitive bins \nllabel{ln:algo_2_on_sn}}

\nl \lElse{Execute Algorithm~\ref{alg:bin_creation}$(S,\mathit{NS})$ \nllabel{ln:resend_shares_for_single_fetch_single}}
}

\caption{An extension to the bin-creation Algorithm~\ref{alg:bin_creation} for the base case, $|S|<|\mathit{NS}|$.}
\label{alg:bin_extension}
}
\end{algorithm}
\setlength{\textfloatsep}{0pt}

\smallskip
\noindent\textbf{Algorithm~\ref{alg:bin_extension} description.} An extension to the bin-creation Algorithm~\ref{alg:bin_creation} is provided in Algorithm~\ref{alg:bin_extension} that handles the case when the number of non-sensitive values ($|S|<|\mathit{NS}|$) is close to a square number.\footnote{The case of $|S| > |\mathit{NS}|$ can be handled by applying Algorithm~\ref{alg:bin_extension} in a reverse way.} Algorithm~\ref{alg:bin_extension} first finds two approximately square factors of non-sensitive values and the cost; Line~\ref{ln:largest_divisors_extension}. Algorithm~\ref{alg:bin_extension} also finds a square number, say $z$, closest to the non-sensitive values and the cost; Line~\ref{ln:find_sn}. Now, Algorithm~\ref{alg:bin_extension} creates bins using a method that results in fewer retrieved tuples (Line~\ref{ln:cost_check}). When Algorithm~\ref{alg:bin_extension} creates bins using the square number closest to the non-sensitive values (Line~\ref{ln:algo_2_on_sn}), the \emph{remaining} non-sensitive values (\textit{i}.\textit{e}., $|\mathit{NS}|-z^2$) can be handled by assigning an equal number of the remaining non-sensitive values in the bins. Note that the sensitive and associated non-sensitive values are assigned to bins in an identical manner as in Algorithm~\ref{alg:bin_creation} (Lines~\ref{ln:sesnitive_allocate}-\ref{ln:assign_remaining_NS}).



}

\subsection{General Case: Multiple Values with Multiple Tuples}
\label{subsec:Multiple Records with Multiple Values}
In this section, we will generalize Algorithms~\ref{alg:bin_creation}-\ref{alg:bin_extension} to consider a case
when different data values have different numbers of associated tuples. First, we will show that sensitive values with different numbers of tuples may provide enough information to the adversary leading to the size, frequency-count attacks, and may disclose some information about the sensitive data. Hence, in the case of multiple values with multiple tuples, Algorithms~\ref{alg:bin_creation}-\ref{alg:bin_extension} cannot be directly implemented. We, thus, develop a strategy to overcome such a situation.

\parskip 0pt
\setlength{\parindent}{15pt}

\smallskip
\noindent
\textbf{Size attack scenario in the base QB.} Consider an assignment of 10 sensitive and 10 non-sensitive values to bins using Algorithm~\ref{alg:bin_creation}; see Figure~\ref{fig:qb}. Assume that a sensitive value, say $s_1$, has 1000 sensitive tuples and an associated non-sensitive value, say $\mathit{ns}_1$, has 2000 tuples, while all the other values have only one tuple each. Further, assume that \emph{each data value represents the salary of employees}.

In this example, consider a query execution for a value, say $\mathit{ns}_1$. The DB owner retrieves tuples from two bins: $\mathit{SB}_1$ (containing encrypted tuples of values $s_1$ and $s_6$) and $\mathit{NSB}_0$ (containing tuples of values $\mathit{ns}_1, \mathit{ns}_2,\mathit{ns}_3,\mathit{ns}_5,\mathit{ns}_{11}$); see Figure~\ref{fig:qb}. Obviously, the number of retrieved tuples satisfying the values of the bins $\mathit{SB}_1$ and $\mathit{NSB}_0$ will be highest (\textit{i}.\textit{e}., 3005) as compared to the number of tuples retrieved based on any two other bins. Thus, the retrieval of the two bins $\mathit{SB}_1$ and $\mathit{NSB}_0$ provides enough information to the adversary to determine which one is the sensitive bin associated with the bin holding the value $\mathit{ns}_1$. Moreover, after observing many queries and having background knowledge, the adversary may estimate that 1000 people in the sensitive relation earn a salary equal to the value $\mathit{ns}_1$.

Thus, in the case of different sensitive values having different numbers of tuples, Algorithm~\ref{alg:bin_creation} cannot satisfy the \emph{second condition of partitioned data security} (\textit{i}.\textit{e}., the adversary is able to distinguish two sensitive values based on the number of retrieved tuples, which was not possible before the query execution, and concludes that a sensitive value ($s_1$ in the above example) has more tuples than any other sensitive value) though preserving all surviving matches, and holding Theorems~\ref{th:algorithm correctness} and~\ref{th:Preserve the perfect data security} to be true.

\emph{In order for the second condition of partitioned data security to hold (and for the scheme to be resilient to the size and frequency-count attacks, as illustrated above), sensitive bins need to hold identical numbers of tuples}. A trivial way of doing this is to outsource some encrypted fake tuples such that the number of tuples in each sensitive bin will be identical. However, we need to be careful; otherwise, adding fake tuples in each sensitive bin may increase the \emph{cost}, if all the heavy-hitter sensitive values are allocated to a single bin. This fact will be clear in the following example.

\begin{figure}[h]
\begin{center}
  \begin{minipage}{.48\linewidth}
  \centering
  \includegraphics[scale=0.7]{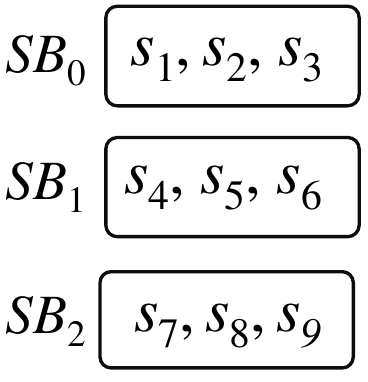}
  \subcaption{The first way.}
  \label{fig:fig_allocation_problem1}
  \end{minipage}
  \begin{minipage}{.49\linewidth}
  \centering
  \includegraphics[scale=0.7]{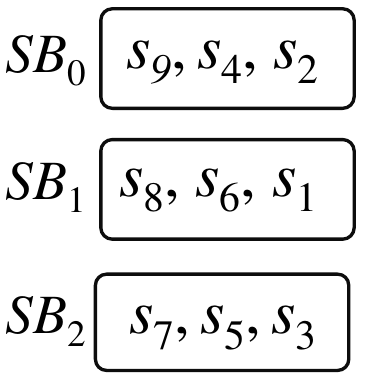}
  \subcaption{The second way.}
  \label{fig:fig_allocation_problem2}
  \end{minipage}
\end{center}
\caption{An assignment of 9 sensitive values to 3 bins.}
\label{fig:allocation_problem}
\end{figure}
\noindent\textnormal{\textbf{Example 6: (Illustrating ways to assign \emph{sensitive} values to bins to minimize the addition of fake tuples).}} Consider 9 sensitive values, say $s_1, s_2, \ldots, s_9$, having 10, 20, 30, 40, 50, 60, 70, 80, and 90 tuples, respectively.\footnote{We assume that there are 9 non-sensitive values, and computed that we need 3 sensitive and 3 non-sensitive bins.} There are multiple ways of assigning these values to three bins so that we need to add a minimum number of fake tuples to each bin. Figure~\ref{fig:allocation_problem} shows two different ways to assign these values to bins. Figure~\ref{fig:fig_allocation_problem2} shows the best way -- to minimize the addition of fake encrypted tuples; hence minimizing the cost. However, bins in Figure~\ref{fig:fig_allocation_problem1} require us to add 180 and 90 fake encrypted tuples to the bins $\mathit{SB}_0$ and $\mathit{SB}_1$, respectively.

Note that there is no need to add any fake tuple if the non-sensitive values have identical numbers of tuples. In that case, the adversary cannot deduce which sensitive bin contains sensitive tuples associated with a non-sensitive value. However, it is obvious that any fake non-sensitive tuple cannot be added in clear-text.

Before describing how to add fake encrypted tuples to bins, we show that a partitioning of sensitive values over $\mathit{SB}$ bins may lead to identical numbers of tuples in each bin, where a bin is not required to hold at most $y$ values, is not a communication-efficient solution. For example, consider 9 sensitive values, where a value, say $s_1$, has 100 tuples and all the other values, say $s_2,s_3,\ldots, s_9$, have 25 tuples each. In this case, we may get bins as shown in Figure~\ref{fig:qb heavy hitter}. Note that the bins $\mathit{SB}_1$ and $\mathit{SB}_2$ are associated with all the three non-sensitive bins while the bin $\mathit{SB}_0$ is associated with only $\mathit{NSB}_0$ (thus, the given bins do not prevent the surviving matches). In order to associate each sensitive bin with each non-sensitive bin (and hence, preventing all the surviving matches), we need to ask fake queries for bins $\langle \mathit{SB}_0, \mathit{NSB}_1\rangle$ and $\langle \mathit{SB}_0, \mathit{NSB}_2\rangle$.

\begin{figure}[h]
  \centering
  \includegraphics[scale=0.7]{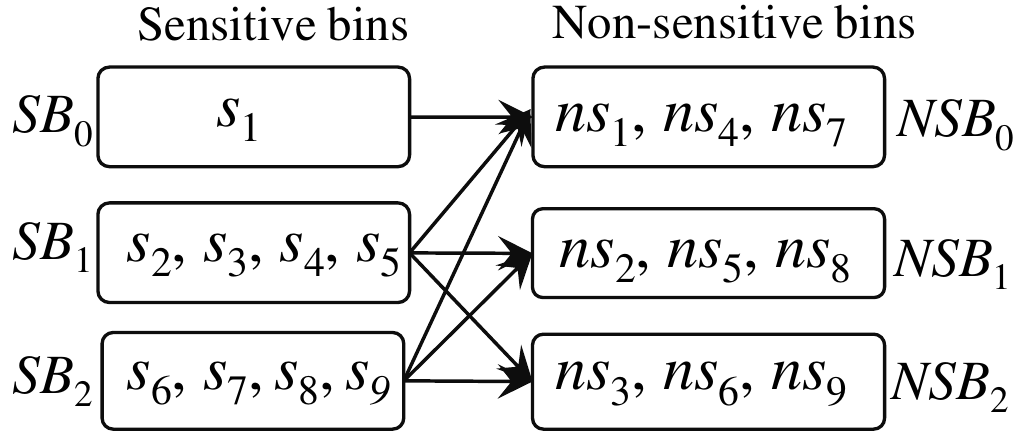}
  \caption{An assignment of a heavy-hitter value but dropping surviving matches.}
\label{fig:qb heavy hitter}
\end{figure}

\smallskip
\noindent
\textbf{Adding fake encrypted tuples.} As an assumption, we know the number of sensitive bins, say $\mathit{SB}$, using Algorithm~\ref{alg:bin_creation} or~\ref{alg:bin_extension}. Here, our objective is to assign sensitive values to bins such that each bin holds identical numbers of tuples while minimizing the number of fake tuples in each bin. To do this, the strategy is given below:
\begin{enumerate}[noitemsep,leftmargin=0.01in]
  \item Sort all the values in a decreasing order of the number of tuples.
  \item Select $\mathit{SB}$ largest values and allocate one in each bin.
  \item Select the next value and find a bin that is containing the fewest number of tuples. If the bin is holding less than $y$ values, then add the value to the bin; otherwise, select another bin with the fewest number of tuples. Repeat this step, for allocating all the values to sensitive bins.
  \item Add fake tuples' values to the bins so that each bin contains identical numbers of tuples.
  \item Allocate non-sensitive values as per Algorithm~\ref{alg:bin_creation} (Lines~\ref{ln:assign_value_to_NS_bucket} and~\ref{ln:assign_remaining_NS}).
\end{enumerate}


{
\section{Other Operations}
\label{sec:Other Operations}

\subsection{Join Queries}
\label{subsec:Join Queries}
Let $R$ be a parent relation that is partitioned into a sensitive relation $R_s$ and a non-sensitive relation $R_{\mathit{ns}}$. Let $S$ be a child relation that is partitioned into a sensitive relation $S_s$ and a non-sensitive relation $S_{\mathit{ns}}$. {  \emph{We assume that a tuple of the relation $R_s$ cannot have any tuple in the child table $S_{\mathit{ns}}$}. In order words, a sensitive tuple with a join key, say $k_i$, of the parent table $R_s$ cannot have a non-sensitive tuple with the joining key $k_i$ in the non-sensitive child table $S_{\mathit{ns}}$.
However, a non-sensitive tuple with a join key, say $k_j$, of the parent table $R_\mathit{ns}$ can have a sensitive tuple with the joining key $k_j$ in the sensitive child table $S_s$. Thus, in the partitioned computing model, the primary-key-to-foreign-key join of $R$ and $S$ is computed as follows:
$$R\bowtie S = (R_s \bowtie S_s) \cup (R_{\mathit{ns}} \bowtie S_{\mathit{ns}}) \cup (R_{\mathit{ns}} \bowtie S_s)$$

Note that our objective is not to build a secure cryptographic technique for joining the sensitive relations. Thus, we use any existing cryptographic technique, \textit{e}.\textit{g}., CryptDB~\cite{DBLP:journals/cacm/PopaRZB12}, SGX-based Opaque~\cite{opaque},~\cite{DBLP:conf/icdt/ArasuK14},~\cite{DBLP:journals/tods/PangD14}, or~\cite{DBLP:conf/dbsec/DolevL016} to join sensitive relations. In addition, our objectives in joining two relations are:
\begin{enumerate}[noitemsep,nolistsep]
  \item Hide which sensitive tuples (of the relation $S_s$) join with a non-sensitive tuple (of the relation $R_{\mathit{ns}}$). For example, we need to hide that $t_2$ of Tables~\ref{fig:s s sensitive Project relation} should join with $r_2$ of Table~\ref{fig:r ns non-sensitive emp}.
  \item Hide which are the encrypted tuples of the output of $(R_s \bowtie S_s) \cup (R_{\mathit{ns}} \bowtie S_s)$ associated with a non-sensitive tuple of $R_{\mathit{ns}} \bowtie S_{\mathit{ns}}$. For example, we need to hide that $r_2$ of Tables~\ref{fig:r ns non-sensitive emp} should join with $t_5$ of Table~\ref{fig:s ns Project relation}.
\end{enumerate}
}

\begin{table*}[!h]
{
\begin{center}
\begin{minipage}[t]{.49\linewidth}
    \centering
    \begin{tabular}{|l|l|l|} \hline
            &  \textbf{EID}  & \textbf{Name}   \\\hline\hline
           $r_1$ & E101 & Adam   \\\hline
           $r_2$ & E102 & Bob  \\\hline
           $r_3$ & E103 & John   \\\hline
    \end{tabular}
    \subcaption{A relation $R =$ \texttt{Employee} relation.}
    \label{fig:employee relation with sensitive and non-sensitive tuples}
\end{minipage}
\begin{minipage}[t]{.49\linewidth}
    \centering
    \begin{tabular}{|l|l|l|l|} \hline
         & \textbf{EeID}  & \textbf{Project Name} \\ \hline\hline
        $t_1$ & E101 & Security  \\ \hline
        $t_2$ & E102 & Design   \\ \hline
        $t_3$ & E103 & Code     \\ \hline
        $t_4$ & E103 & Sale   \\ \hline
        $t_5$ & E102 & Sale   \\ \hline
    \end{tabular}
    \subcaption{A relation $S=$ \texttt{Project} relation.}
    \label{fig:Project relation with sensitive and non-sensitive tuples}
\end{minipage}
\end{center}
\caption{Two relations with their sensitive and non-sensitive tuples.}
\label{fig:Two relations with their sensitive and non-sensitive tuples}
}
\end{table*}

\begin{table*}[!h]
{
\begin{center}
\begin{minipage}[t]{.2\linewidth}
    \centering
    \begin{tabular}{|l|l|l|} \hline
            &  \textbf{EID}  & \textbf{Name}   \\\hline\hline
           $r_1$ & E101 & Adam   \\\hline
    \end{tabular}
    \subcaption{$R_s$.}
    \label{fig:rs sensitive emp}
\end{minipage}
\begin{minipage}[t]{.20\linewidth}
    \centering
    \begin{tabular}{|l|l|l|} \hline
            &  \textbf{EID}  & \textbf{Name}   \\\hline\hline
           $r_2$ & E102 & Bob  \\\hline
           $r_3$ & E103 & John   \\\hline
    \end{tabular}
    \subcaption{$R_{\mathit{ns}}$.}
    \label{fig:r ns non-sensitive emp}
\end{minipage}
\begin{minipage}[t]{.28\linewidth}
    \centering
    \begin{tabular}{|l|l|l|l|} \hline
         & \textbf{EeID}  & \textbf{Project Name} \\ \hline\hline
        $t_1$ & E101 & Security  \\ \hline
        $t_2$ & E102 & Design   \\ \hline
    \end{tabular}
    \subcaption{$S_s$.}
    \label{fig:s s sensitive Project relation}
\end{minipage}
\begin{minipage}[t]{.254\linewidth}
    \centering
    \begin{tabular}{|l|l|l|l|} \hline
         & \textbf{EeID}  & \textbf{Project Name} \\ \hline\hline
        $t_3$ & E103 & Code     \\ \hline
        $t_4$ & E103 & Sale   \\ \hline
        $t_5$ & E102 & Sale   \\ \hline
    \end{tabular}
    \subcaption{$S_{\mathit{ns}}$.}
    \label{fig:s ns Project relation}
\end{minipage}
\end{center}
\caption{Sensitive and non-sensitive relations created from two relations of Table~\ref{fig:Two relations with their sensitive and non-sensitive tuples}.}
\label{fig:Sensitive and non-sensitive relations created from two relations of Table}
}
\end{table*}

\begin{table*}[!h]
{
\begin{center}
\begin{minipage}[t]{.49\linewidth}
    \centering
    \begin{tabular}{|l|l|l|} \hline
            &  \textbf{EID}  & \textbf{Name}   \\\hline\hline
           $r_1$ & E101 & Adam   \\\hline
           $r_2$ & E102 & Bob  \\\hline
    \end{tabular}
    \subcaption{$R_{\mathit{ps}}$.}
    \label{fig:r ps}
\end{minipage}
\begin{minipage}[t]{.49\linewidth}
    \centering
    \begin{tabular}{|l|l|l|} \hline
            &  \textbf{EID}  & \textbf{Name}   \\\hline\hline
           $r_2$ & E102 & Bob  \\\hline
           $r_3$ & E103 & John   \\\hline
    \end{tabular}
    \subcaption{$R_{\mathit{ns}}$ same as Table~\ref{fig:r ns non-sensitive emp}.}
    \label{fig:r ns}
\end{minipage}
\end{center}
\caption{Sensitive relation with pseudosensitive tuples and non-sensitive relation, created from $R_s =$ \texttt{Employee} relation of Table~\ref{fig:employee relation with sensitive and non-sensitive tuples}.}
\label{fig:pseudoSensitive and non-sensitive relations created from two relations of Table}
}
\end{table*}


\smallskip
\noindent\textbf{The DB owner-side.}
In order to join, the relations $R_{\mathit{ns}}$ and $S_s$, we follow the approach given in~\cite{TR} that pre-computes all the tuples of $R_{\mathit{ns}}$ that join with $S_s$. We call all such tuples of $R_{\mathit{ns}}$ as pseudo-sensitive tuples. In~\cite{TR}, the authors found that the size of pseudo-sensitive data does not need to consider the entire $R_{\mathit{ns}}$ as sensitive. Particularly, at 10\% of sensitivity level, pseudo-sensitive data is only a fraction (25\%) of the entire database.

In our join strategy (see Algorithm~\ref{alg:join_processing}), before outsourcing the relations $R$ and $S$, the DB owner finds pseudo-sensitive tuples of $R_{\mathit{ns}}$ and keeps them with sensitive tuples of $R_s$, resulting in a new relation, denoted by $R_{\mathit{ps}}$, containing sensitive and pseudo-sensitive tuples (Line~\ref{ln:gen_pseudo_sensitive} of Algorithm~\ref{alg:join_processing}). Now, the DB owner outsources (\textit{i}) encrypted relations $R_{\mathit{ps}}$ and $S_s$, and (\textit{ii}) cleartext relations $R_{\mathit{ns}}$, $S_{\mathit{ns}}$ (Line~\ref{ln:outsource_join_algo} of Algorithm~\ref{alg:join_processing}). Additionally, the DB owner maintains the information for bin-creation (Algorithm~\ref{alg:bin_creation}), which can be used to retrieve tuples after join. Thus, in our case, the join of $R$ and $S$ is converted into the following join:
$$R\bowtie S = (R_{\mathit{ns}} \bowtie S_{\mathit{ns}}) \cup (R_{\mathit{ps}} \bowtie S_s)$$

\smallskip
\noindent\textbf{The cloud-side.}
We use any cryptographic technique for $R_{\mathit{ps}} \bowtie S_s$, and of course, join of the relations $R_{\mathit{ns}}$ and $S_{\mathit{ns}}$ is carried out in the cleartext.

\smallskip
\noindent\textbf{Note: non-foreign-key joins.} The above strategy can also be extended to non-foreign-key joins by encrypting pseudo-sensitive tuples of $S_{\mathit{ns}}$ with $S_s$. However, in this case, we need to avoid join of pseudo-sensitive tuples of $R_{\mathit{ns}}$ and $S_{\mathit{ns}}$ in the encrypted domain, since these tuples will also join in cleartext. It can be done if the DB owner can add an attribute to each sensitive relation to mark such pseudo-sensitive tuples.
}

\DontPrintSemicolon
\LinesNotNumbered \begin{algorithm}[!t]
{
\textbf{Inputs:} Two relations: $R(\mathit{key}, A_1, A_2, \cdots A_m)$ and $S(\mathit{key}, B_1,B_2, \cdots B_{m^{\prime}}) $

\textbf{Outputs:} $R \bowtie S$

\underline{\textbf{DB owner}} \\

\nl Create $R_s$, $R_{ns}$, $S_s$, and $S_{ns}$\nllabel{ln:partition_relation_join_algo}

\nl $\mathit{pseudo\_sensitive\_key}[] \leftarrow \{key \in R_{\mathit{ns}} \mid \Pi_{key} (R_{\mathit{ns}}) \cap  \Pi_{\mathit{key}} (S_{s}) \neq \emptyset \}$ {\scriptsize {\scriptsize \tcp{Retrieve all the keys in $R_{\mathit{ns}}$ that joins with the relation $S_s$}}}  \nllabel{ln:gen_pseudo_sensitive_keys}

\nl $R_{\mathit{ps}} \leftarrow R_{s} \cup (\sigma_{\mathit{key} \in \mathit{pseudo\_sensitive\_key}[]} (R_{ns})$) {\scriptsize  \tcp{Retrieving tuples from $R_{\mathit{ns}}$ based on the the keys present in $\mathit{pseudo\_sensitive\_key[]}$ and merging them with the relation $R_s$}} \nllabel{ln:gen_pseudo_sensitive}

\nl Encrypt $R_{ps}$ and $S_{s}$

\nl Outsource encrypted $R_{\mathit{ps}}$, encrypted $S_{s}$, cleartext $R_{\mathit{ns}}$, and cleartext $S_{\mathit{ns}}$ to cloud \nllabel{ln:outsource_join_algo}

\underline{\textbf{Cloud}} \\

\nl $\mathit{join}_{\mathit{sensitive\_output}} \leftarrow R_{\mathit{ps}} \bowtie S_{s}$, $\mathit{join}_{\mathit{non\_sensitive\_output}} \leftarrow R_{\mathit{ns}} \bowtie S_{\mathit{ns}}$  \nllabel{ln:join_output}

\underline{\textbf{DB owner}} \\
\nl \lIf{$\sigma_{A_i=s_j}(\mathit{join}_{\mathit{sensitive\_output}}, \mathit{join}_{\mathit{non\_sensitive\_output}})$}
{
     \text{Execute Algorithm~\ref{alg:bin_retrieval}  {\scriptsize \tcp{If the user is interested to fetch a tuple having $s_j$ in attribute $A_i$}}}
}
\caption{Algorithm for execution join queries.}
\label{alg:join_processing}
}
\end{algorithm}
\setlength{\textfloatsep}{0pt}

{
\subsection{Range Queries}
\label{subsec:range Queries}
Let $A$ be an attribute on which we want to execute a range query. For answering a range query, we convert it into the selection query, which can be executed using QB. However, a careless execution of QB for answering a range query, which is converted into selection queries, may result in retrieval of either entire sensitive or non-sensitive data. For example, consider 16 sensitive values, say $s_1, s_2, \ldots s_{16}$, and their associated non-sensitive values, say $\mathit{ns}_1,\mathit{ns}_2,\ldots,\mathit{ns}_{16}$, where the sensitive value $s_i$ is associated with the non-sensitive value $\mathit{ns}_i$. Figure~\ref{fig:allocating 16 values for range query} shows a way to assign these values to bins.

\begin{figure}[!h]
  \centering
  \includegraphics[scale=.6]{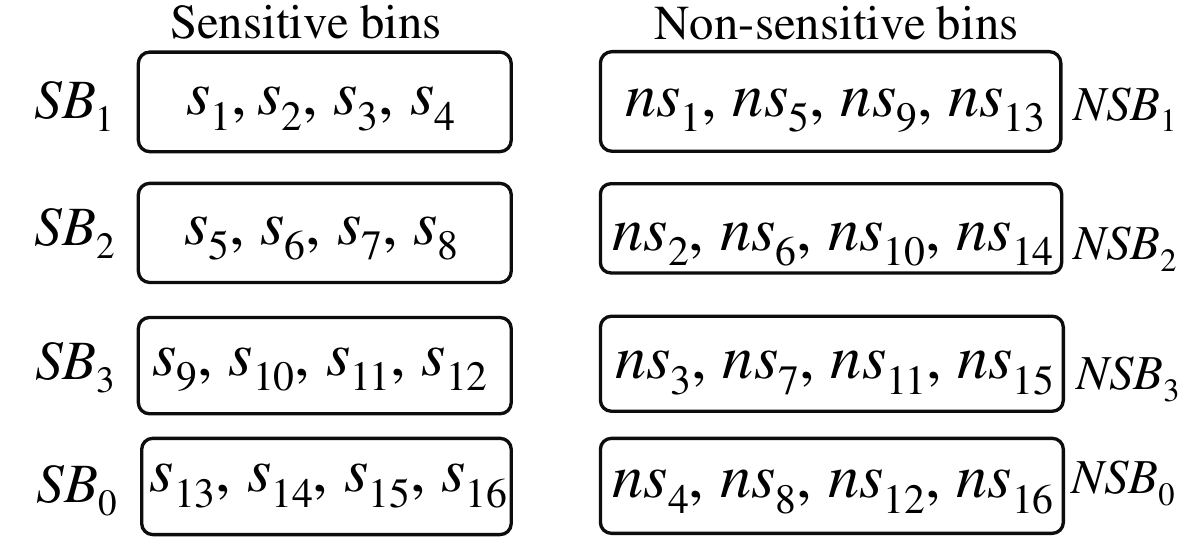}
  \caption{A way to allocate 16 sensitive and non-sensitive values to bins by following Algorithm~\ref{alg:bin_creation}.}
  \label{fig:allocating 16 values for range query}
  \FFF\FFF
\end{figure}

Consider a range query for values $s_1$ to $s_4$. Answering this range query using QB will result in retrieval of the entire non-sensitive data and the bin $\mathit{SB}_1$. Our objective is to create bins in a way that results in a few tuple retrieval.

We describe a procedure for the case $|S|\leq |\mathit{NS}|$, as the restriction is followed by Algorithm~\ref{alg:bin_creation} in \S\ref{subsec:The Base Case}. We use the example of 16 sensitive and 16 non-sensitive values of the attribute $A$. In order to answer range queries, the DB owner builds a full binary tree on the unique values of the attribute $A$ of the non-sensitive relation and traverses the tree to find a node that covers the range. Thus, the DB owner retrieves tuples satisfying a larger range query that also covers the desired range query. Note that many papers~\cite{DBLP:conf/icde/HacigumusMI02,DBLP:journals/pvldb/LiLWB14,DBLP:journals/ton/LiLWB16,DBLP:journals/tods/DemertzisPPDGP18} used the same approach of fetching a large range value to satisfy the desired range value, and hence, preventing exact range values to be revealed to the adversary.

\medskip
\noindent\textbf{Full binary tree and bin creation.} The DB owner first builds a full binary tree\footnote{The DB owner may also build a $k$-ary tree, where each node (except leaf nodes) contains $k$ child nodes.} for the values of the attribute $A$ of the non-sensitive relation; see Figure~\ref{fig:full binary tree} for 16 non-sensitive values.
For each level of the tree, the DB owner applies Algorithm~\ref{alg:bin_creation} that takes nodes of the level as inputs. In particular, for the leaf nodes, \textit{i}.\textit{e}., level 0, Algorithm~\ref{alg:bin_creation} takes 16 non-sensitive values, and produces 4 sensitive and 4 non-sensitive bins, by following Lines~\ref{ln:largest_divisors}-\ref{ln:assign_remaining_NS} of Algorithm~\ref{alg:bin_creation}.

\begin{figure}[!h]
  \centering
  \includegraphics[scale=.45]{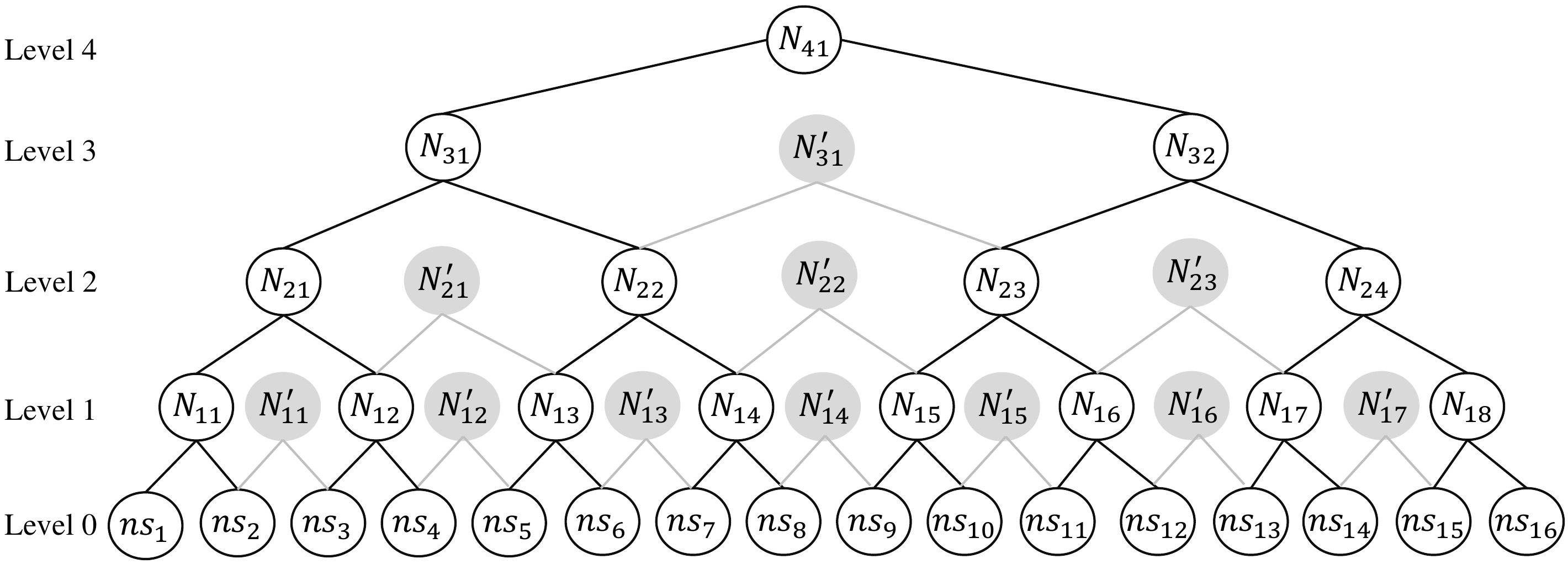}
  \caption{A full binary tree with some additional nodes for 16 non-sensitive values.}
  \label{fig:full binary tree}
\end{figure}

At the level 1, Algorithm~\ref{alg:bin_creation} takes 8 inputs that represent the nodes ($N_{11},N_{12},\ldots,N_{18}$; see white nodes in Figure~\ref{fig:full binary tree}) at the level 1, and each input value of the level 1 holds two non-sensitive values, which are child nodes of a level 1's node. For example, the node $N_{11}$ holds two values $\mathit{ns}_1,\mathit{ns}_2$. For the 8 values, Algorithm~\ref{alg:bin_creation} provides two non-sensitive bins (each is containing 8 values) and four sensitive bins (each is containing 4 values). Let $\mathit{NSB}_{ij}$ be the $j^{\mathit{th}}$ non-sensitive bin at the $i^{\mathit{th}}$ level, and let $\mathit{SB}_{ij}$ be the $j^{\mathit{th}}$ sensitive bin at the $i^{\mathit{th}}$ level. Thus, Algorithm~\ref{alg:bin_creation} produces the following bins:

\begin{center}
$\mathit{NSB}_{10}$ containing $\langle N_{11},N_{12},\ldots,N_{14}\rangle$,

$\mathit{NSB}_{11}$ containing $\langle N_{15},N_{16},\ldots,N_{18}\rangle$,

$\mathit{SB}_{10}$ containing $\langle s_1,s_2,s_9,s_{10}\rangle$,

$\mathit{SB}_{11}$ containing $\langle s_3,s_4,s_{11},s_{12}\rangle$,

$\mathit{SB}_{12}$ containing $\langle s_5,s_6,s_{13},s_{14}\rangle$,

$\mathit{SB}_{13}$ containing $\langle s_7,s_8,s_{15},s_{16}\rangle$.
\end{center}

At level 2, Algorithm~\ref{alg:bin_creation} takes 4 inputs that represent the nodes ($N_{21},N_{22},N_{23},N_{24}$; see white nodes in Figure~\ref{fig:full binary tree}) at the level 2, and each input value of the level 2 holds four non-sensitive values, which are grandchild nodes of a level 2's node. For the 4 input values, Algorithm~\ref{alg:bin_creation} provides two non-sensitive bins (each is containing 8 values) and two sensitive bins (each is containing 8 values), as follows:

\begin{center}
$\mathit{NSB}_{20}$ containing $\langle N_{21},N_{22}\rangle$,

$\mathit{NSB}_{21}$ containing $\langle N_{23},N_{24}\rangle$,

$\mathit{SB}_{20}$ containing $s_5,s_6,s_7,s_8,s_{9},s_{10},s_{11},s_{12}$,

$\mathit{SB}_{21}$ containing $s_1,s_2,s_3,s_4,s_{13},s_{14},s_{15},s_{16}$.
\end{center}
The DB owner follows the same procedure for the higher nodes, except the root node and child nodes of the root node.

Further, at each level except the root node, the child nodes of the root node, and the leaf nodes, the DB owner creates additional nodes (see gray-colored nodes in Figure~\ref{fig:full binary tree}) that become parent nodes of the lower level's two adjacent nodes that do not have a common parent node. The algorithm given in~\cite{DBLP:journals/tods/DemertzisPPDGP18} also uses these additional nodes for answering a range query. For example, at the level 2 in Figure~\ref{fig:full binary tree}, the DB owner creates 7 such nodes. Let $\mathit{NSB}_{ij}^{\prime}$ be the $j^{\mathit{th}}$ non-sensitive bin at the $i^{\mathit{th}}$ level for these additional nodes, and let $\mathit{SB}_{ij}^{\prime}$ be the $j^{\mathit{th}}$ sensitive bin at the $i^{\mathit{th}}$ level for these additional nodes. Algorithm~\ref{alg:bin_creation} takes these 7 inputs and produces 2 non-sensitive bins (each is containing 8 values) and 4 sensitive bins (each is containing 8 values), as follows:\footnote{Note that the bins $\mathit{NSB}_{11}^{\prime}$ and $\mathit{SB}_{13}^{\prime}$ will ask to fetch two fake tuples each to maintain an identical-sized bin.}

\begin{center}
$\mathit{NSB}_{10}^{\prime}$ containing $\langle N_{11}^{\prime},N_{12}^{\prime},\ldots,N_{14}^{\prime}\rangle$,

$\mathit{NSB}_{11}^{\prime}$ containing $\langle N_{15}^{\prime},N_{16}^{\prime},N_{17}^{\prime},\textnormal{ 2 fake tuples}\rangle$,

$\mathit{SB}_{10}^{\prime}$ containing $\langle s_2,s_3,s_{10},s_{11}\rangle$,

$\mathit{SB}_{11}^{\prime}$ containing $\langle s_4,s_5,s_{12},s_{13}\rangle$,

$\mathit{SB}_{12}^{\prime}$ containing $\langle s_6,s_7,s_{14},s_{15}\rangle$,

$\mathit{SB}_{13}^{\prime}$ containing $\langle s_8,s_9,\textnormal{ 2 fake tuples}\rangle$.
\end{center}


\medskip
\noindent
\textbf{Note.}  {The self-explainable pseudocode in Algorithm~\ref{alg:bin_creation_range}, given in Appendix~\ref{app_sec:Pseudocode}, describes all the above steps of binary-tree creation and bin-creation for range queries.}

\medskip
\noindent\textbf{Bin retrieval and answering range queries.} We provide two approaches: best-match method and least-match method, for retrieving the bins in answering a range query.

\smallskip\noindent\emph{Best-match method}. This method traverses the tree in a bottom-up fashion and finds a node that covers the entire range. Then, it retrieves a non-sensitive bin corresponding to this node and a sensitive bin. 
For example, if the query is for values $\mathit{ns}_1$ to $\mathit{ns}_4$, then by traversing the tree (see Figure~\ref{fig:full binary tree}) in a bottom-up fashion, the DB owner retrieves a non-sensitive bin corresponding to the level 2, since the node $N_{21}$ covers the entire range. Thus, the DB owner retrieves the bins $\mathit{NSB}_{20}$  and $\mathit{SB}_{21}$. \emph{ {The self-explainable pseudocode in Algorithm~\ref{alg:bin_retrieve_range_best_match}, given in Appendix~\ref{app_sec:Pseudocode}, describes the best-match method for answering range queries.}}

\smallskip\noindent\emph{Least-match method}.  Assume a query is for values $\mathit{ns}_8$ to $\mathit{ns}_{12}$. The best-match method will find only the root node that satisfies this query, and hence, it will result in the retrieval of entire non-sensitive or sensitive relation. Thus, we propose a different method that breaks the range query into many sub-range queries and finds a minimal set of nodes that cover the range.
For example, the node $N_{23}$ satisfies the query for value $\mathit{ns}_9$ to $\mathit{ns}_{12}$, and the leaf node having the value 8 satisfies the query for the value $\mathit{ns}_8$. Thus, the DB owner retrieves the bins $\mathit{NSB}_{21}$  and $\mathit{SB}_{20}$ to satisfy the query for the value $\mathit{ns}_9$ to $\mathit{ns}_{12}$, and a sensitive bin and a non-sensitive bin to satisfy the query for the value $\mathit{ns}_8$.

\smallskip\noindent\textbf{Aside: using additional nodes for answering a range query by following Algorithm~\ref{alg:bin_retrieval}.} Assume a query is for values $\mathit{ns}_4$ to $\mathit{ns}_7$. The best-match method finds only the root node that satisfies the query, and hence, results in retrieval of the entire non-sensitive or sensitive relation. In contrast, the least-match method will break the query into sub-range queries, such as (\textit{Q1}) a query for $\mathit{ns}_4$, a query for $\mathit{ns}_5,\mathit{ns}_6$, and a query for $\mathit{ns}_7$, or (\textit{Q2}) four selection queries one for each value.

The first query (\textit{Q1}) will find the node $N_{13}$ that covers the values $\mathit{ns}_5,\mathit{ns}_6$, and two leaf nodes one for $\mathit{ns}_4$ and another for $\mathit{ns}_7$. This will result in retrieval of 28 tuples, such as one sensitive bin and non-sensitive bin for $s_4$ (containing 4 tuples in each; see Figure~\ref{fig:allocating 16 values for range query}), one sensitive bin and non-sensitive bin for $s_7$ (containing 4 tuples in each; see Figure~\ref{fig:allocating 16 values for range query}), and the bin $\mathit{NSB}_{10}$ (containing 8 tuples) and $\mathit{SB}_{12}$ (containing 4 tuples) for answering the query for a range $\mathit{ns}_5$ to $\mathit{ns}_6$. However, the second query (\textit{Q2}) will be worse in terms of retrieving the tuples. It will result in the retrieval of the entire non-sensitive data (see Figure~\ref{fig:allocating 16 values for range query}).

In order to reduce the number of retrieved tuples, the DB owner can use the bins for the additional nodes. In particular, the DB owner finds that the nodes $N^{\prime}_{12}$ and $N^{\prime}_{13}$ that satisfy the value $\mathit{ns}_4$-$\mathit{ns}_5$ and $\mathit{ns}_6$-$\mathit{ns}_7$, respectively. Thus, the bins $\mathit{NSB}_{10}^{\prime}$ (containing 8 tuples), $\mathit{SB}_{11}^{\prime}$ (containing 4 tuples), $\mathit{SB}_{12}^{\prime}$ (containing 4 tuples) can fulfill the query, and will result in retrieval of 16 tuples.

Note that by using the bins for the additional nodes, one can answer queries for two adjacent nodes that do not share a common parent in the original full binary tree, for example, values 8 and 9.


\subsection{Insert Operation and Re-binning}
\label{subsec:Insert Operation and Re-binning}
QB does not allow outsourcing new tuples immediately as the new tuples arrive at the DB owner. The DB owner collects enough number of tuples before outsourcing them, while the DB owner can either update the existing bins (by increasing an identical size of each bin) or create all the new bins.

Particularly, the DB owner waits for new tuples until the DB owner collects new sensitive and non-sensitive values equals to the number of existing sensitive and non-sensitive bins, such that each bin receives a new value.\footnote{In case, if the DB owner cannot collect new values equal to the number of bins, the DB owner may outsource fake values.} Let $p$ and $q$ be the number of existing sensitive and non-sensitive bins, respectively. Note that when collecting $p$ sensitive and $q$ non-sensitive values, the DB owner does not outsource these values if they will not become a part of each existing sensitive or non-sensitive bin. Note that if the new values become a part of only one sensitive and one non-sensitive bin, it reveals an association of values. {Further, note that in outsourcing new tuples, the DB owner sends appropriately encrypted sensitive data and cleartext non-sensitive data to the cloud. Thus, the cloud does not have access to any sensitive data in cleartext, and hence, cannot launch an attack, if it is an honest-but-curious adversary, \textit{i}.\textit{e}., a passive attacker, as we mentioned in \S\ref{subsec:Adversarial Model}.}

However, the insertion of more values in existing bins incurs the overhead,  as will be shown in Experiment 6 in~\S\ref{sec:Experimental Validation}. Hence, Algorithm~\ref{alg:bin_creation} is re-executed when the overhead crosses a user-defined threshold.

Now, we describe a procedure for outsourcing new tuples while using the existing sensitive and non-sensitive bins. Let $s_i$ and $\mathit{ns}_j$ be the value of new sensitive and non-sensitive tuples, respectively. When inserting new tuples, the value $s_i$ or $\mathit{ns}_j$ may exist in the outsourced data, and based on the existence of the values we classify them into four groups, as follows: (\textit{i}) \emph{old sensitive value} (old-S): the value $s_i$ exists in the outsourced sensitive data, (\textit{ii}) \emph{new sensitive value} (new-S): the value $s_i$ does not exist in the outsourced sensitive data, (\textit{iii}) \emph{old non-sensitive value} (old-NS): the value $\mathit{ns}_j$ exists in the outsourced non-sensitive data, and (\textit{iv}) \emph{new non-sensitive value} (new-NS): the value $\mathit{ns}_j$ does not exist in the outsourced non-sensitive data.

Based on the above-mentioned four types of values, the following four possible insert scenarios are allowed while using QB.
\begin{enumerate}[noitemsep,leftmargin=0.01in]
  \item \emph{Inserting old-S and old-NS}. This scenario is trivial to handle and does not require any update to the existing bins. The DB owner outsources the encrypted sensitive tuples and non-sensitive data in cleartext.

  \item \emph{Inserting new-S and new-NS}. The DB owner increases the size of each bin by one and inserts the values into existing bins.

  \item \emph{Inserting old-S and new-NS}. Inserting tuples of $s_i$ does not require an update to the sensitive bins. The DB owner checks whether the value $\mathit{ns}_j$ has an associated sensitive value or not in the outsourced data, by following Line 6-7 of Algorithm~\ref{alg:bin_retrieval}. If the value $\mathit{ns}_j$ has an associated sensitive value, say $s_k$, then the DB owner updates the non-sensitive bin associated with a sensitive bin holding $s_k$ with the value $\mathit{ns}_j$, according to Line 6 of Algorithm~\ref{alg:bin_creation}. If the value $\mathit{ns}_j$ has no associated sensitive value, then the DB owner randomly inserts each non-sensitive value, one per non-sensitive bin.

  \item \emph{Inserting new-S and old-NS}. This case is just the opposite of the previous case.
\end{enumerate}
Note that all the four scenarios may require to outsource some fake tuples to have identical-sized sensitive bins. The update/delete operation can also be done as an insert operation, where some additional tuples are outsourced to notify the non-existence of tuples.

\subsection{Conjunctive Queries}
\label{subsec:Conjunctive Queries}
As defined, QB only works for selection queries with a single attribute in the search clause. Conjunctive queries that contain several such conjuncts can also be supported in several ways. First, note that QB can be applied to multiple attributes, say $A$ and $B$, in a relation. During query processing, if a query refers to both attributes $A$ and $B$, we can select the more selective index and execute QB on it without inference attacks. Using QB on both attributes simultaneously, however, unless done carefully, can lead to leakage.
An approach to apply QB is to consider attributes that
appear commonly together in queries as a single (paired) attribute. Thus, values of this
paired attribute would be attribute value pairs on which QB can be applied. In general, the relation scheme will need to be partitioned into attribute subsets on which QB can be applied. During query execution, the query processing algorithm will choose the corresponding attribute subset that is most beneficial (will result in the least overhead) to execute the query. One such solution is to create partitions of singleton attributes, but, then, conjunctive queries will run on a single attribute and reduce to the first solution.

\subsection{Handling Workload-skew Attack}
\label{subsec:Handling Workload-skew Attack}

{
The query execution and the corresponding accessed tuples (in the absence of an access-pattern-hiding scheme) allows an adversary to deduce the frequency of queries, without knowing the cleartext value of the query keyword. Such information coupled with background knowledge may reveal additional information to the adversary. For example, in Table~\ref{fig:employee relation}, if many queries access tuples $t_2$ and $t_4$, then the adversary may deduce that these two tuples may have identical values in one or more attributes. Furthermore, if the adversary has background knowledge that \texttt{John} is the company's CEO (for which people ask many queries), then the adversary may deduce that $t_2$ and $t_4$ belong to \texttt{John}. We call such an attack as \emph{workload-skew attack}. We illustrate the workload-skew based attack (see Figure~\ref{fig:workload_problem1}) and also our approach for addressing it for the base case of QB. In our proposed solution, we consider that a query predicate is either highly frequent or infrequent. Our solution can be  extended to the general case of different frequency query predicates by creating groups of query keywords based on their query frequencies and allocating such groups to non-sensitive bins such that the number of values in each bin is equal.}

\parskip 0pt
\setlength{\parindent}{15pt}


\begin{figure}[h]
\begin{center}
  \begin{minipage}{.49\linewidth}
  \centering
  \includegraphics[scale=0.7]{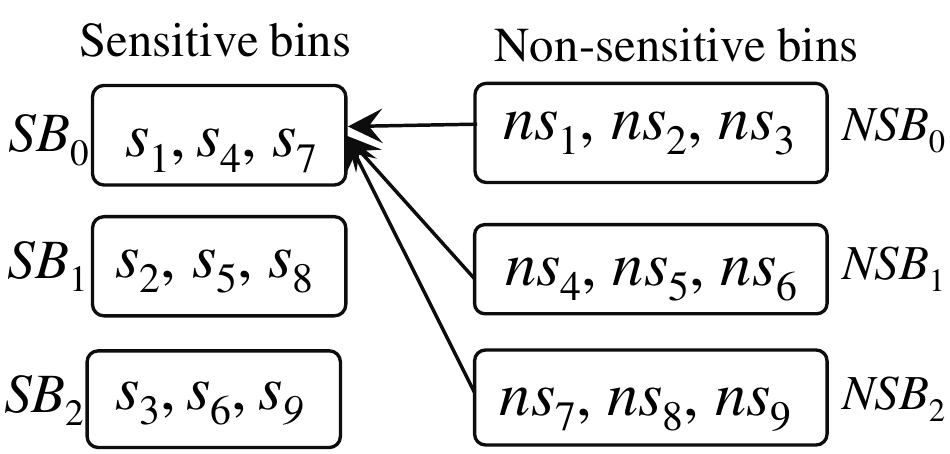}
  \subcaption{The workload-skew attack.}
  \label{fig:workload_problem1}
  \end{minipage}
  \begin{minipage}{.49\linewidth}
  \centering
  \includegraphics[scale=0.7]{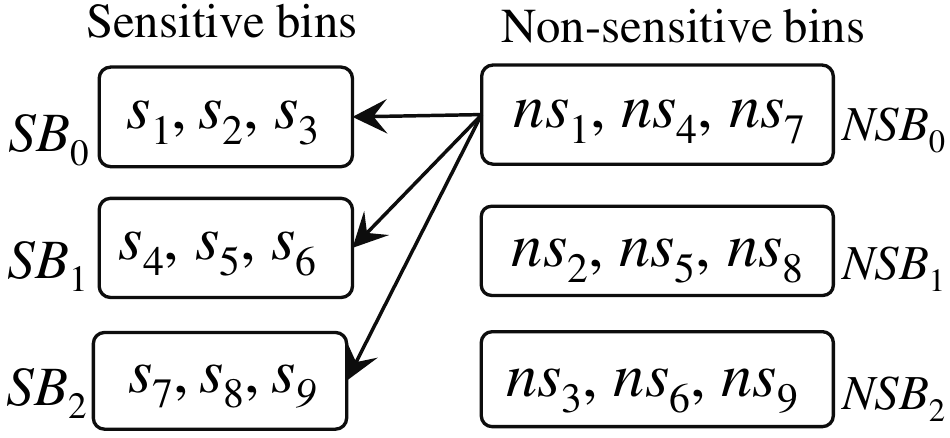}
  \subcaption{A solution.}
  \label{fig:workload_problem2}
  \end{minipage}
\end{center}
\caption{The workload-skew attack and solution under QB, where $\mathit{ns}_1$, $\mathit{ns}_4$, and $\mathit{ns}_7$ are frequent predicates.}
\label{fig:workload_problem}
\end{figure}

Figure~\ref{fig:workload_problem1} shows bins created by Algorithm~\ref{alg:bin_creation} for 9 sensitive values and their associated 9 non-sensitive values. Consider the values $\mathit{ns}_1$, $\mathit{ns}_4$, and $\mathit{ns}_7$ occur most frequently in the query workload. Hence, in this example, the adversary can trivially figure out by observing the sensitive tuple retrieval that only the bin $\mathit{SB}_0$ has the associated sensitive values with $\mathit{ns}_1$, $\mathit{ns}_4$, and $\mathit{ns}_7$. The reason is that these four bins are retrieved more frequently compared to any other bin. Thus, the adversary can determine that the encrypted values $s_1$, $s_4$, and $s_7$ are associated with either $\mathit{ns}_1$, $\mathit{ns}_4$, or $\mathit{ns}_7$. This is more information than what the adversary had prior to the query execution since each sensitive value, \textit{e}.\textit{g}., $s_1$, could be any of the 9 non-sensitive values. However, it is hard for the adversary to find out which sensitive value out of the three sensitive values of the bin $\mathit{SB}_0$ is exactly associated with $\mathit{ns}_1$, $\mathit{ns}_4$, or $\mathit{ns}_7$.\footnote{We are not assuming that a sensitive bin is not associated with each non-sensitive bin. But, because of heavy-hitter queries, the other bins are retrieved less frequently than the bins having frequent selection predicates.} In order to prevent the workload-skew attack, we need to allocate sensitive values carefully, thereby the sensitive values associated with frequent selection predicates are distributed over all the bins.

{
Below, we provide a strategy for handling the workload-skew attack in QB.
The idea is to find groups of either $x (=|\mathit{NSB}|)$ non-sensitive query keywords or $y=|\mathit{SB}|$ sensitive query keywords that are queried almost equally, and then follow the given steps to create bins appropriately.  {The self-explainable pseudocode in Algorithm~\ref{alg:bin_creation_workload}, given in Appendix~\ref{app_sec:Pseudocode}, and the following steps describe bin-creation method for the case of $x$ frequent non-sensitive keywords. The similar method can be used for bin creation in the case of $y$ frequent sensitive keywords.}


\smallskip
\noindent\textbf{Step 1: Bin-creation to deal with workload-skew attack.} In order to deal with workload-skew attack, we, first, need to modify Algorithm~\ref{alg:bin_creation} for bin-creation. The modified bin-creation algorithm (Algorithm~\ref{alg:bin_creation_workload}, given in Appendix~\ref{app_sec:Pseudocode}) contains the following steps:

\begin{enumerate}[noitemsep,leftmargin=0.01in]
\item \emph{Create bins}. Find two largest divisors, say $x\geq y$, of $|\mathit{NS}|$, create $\mathit{NSB}=\lceil |\mathit{NS}|/x\rceil$ non-sensitive bins, and $x$ sensitive bins (Line~\ref{ln:largest_divisors} of Algorithm~\ref{alg:bin_creation}).

\item \emph{Assign non-sensitive values}. Create groups, each of size $x$, of the frequent non-sensitive predicates, resulting in $u\leq \mathit{NSB}$ groups. Assign the $i^{\mathit{th}}$ group to the $i^{\mathit{th}}$ non-sensitive bin. Now, assign all the remaining non-sensitive values, if any, as follows: Find those non-sensitive bins that are not full, \textit{i}.\textit{e}., having less than $x$ non-sensitive values. Then, assign the remaining non-sensitive values to all such non-sensitive bins, such that each non-sensitive bin should contain at most $x$ non-sensitive values.

\item \emph{Assign sensitive values}. Assign the sensitive values associated with a non-sensitive value, say $\mathit{ns}_j = \mathit{NSB}_z[j]$, where $0\leq j\leq x-1$, to the $j^{\mathit{th}}$ sensitive bin at the $z^{\mathit{th}}$ position.
\end{enumerate}

By following the above steps, Figure~\ref{fig:workload_problem2} shows 3 sensitive bins in the case of $\mathit{ns}_1$, $\mathit{ns}_4$, and $\mathit{ns}_7$ as the frequent query predicates. Note that the execution of QB using this strategy insists on retrieving all the sensitive bins for answering frequent predicates. Thus, the adversary cannot determine which bin has a sensitive value associated with the values $\mathit{ns}_1$, $\mathit{ns}_4$, or $\mathit{ns}_7$.\footnote{If there are less then $x$ --- which is the size of a non-sensitive bin --- frequent predicates in a non-sensitive bin, then we need to send fake queries for infrequent keyword of the bin, as frequent as frequent predicates, leading to retrieval of each sensitive bin. It will hide that the how many keywords are frequent and infrequent in the bin.}
}

\smallskip
\noindent\textbf{Step 2: Bin-retrieval --- answering queries.} To retrieve tuples satisfying a query keyword, we use Algorithm~\ref{alg:bin_retrieval}.

}

\subsection{Enhancing Security-Levels of Indexable Techniques}
\label{subsec:Enhancing security-levels of indexable techniques}
We show how QB can be integrated with an indexable cryptographic technique, namely Arx~\cite{arx-popa-2017} that uses a non-deterministic encryption mechanism. In Arx, the DB owner stores each domain value $v$ and the frequency of $v$ in the database. The technique encrypts the $i^{\mathit{th}}$ occurrence of $v$ as a concatenated string $\langle v,i\rangle$ thereby ensuring that no two occurrences of $v$ result in an identical ciphertext. Such a ciphertext representation can then be indexed on the cloud-side. During retrieval, the user keeps track of the histogram of occurrences for each value and generates appropriate ciphertexts that can be used to query the index on the cloud. It is not difficult to see that Arx, by itself, is susceptible to the size, frequency-count, workload-skew, and access-pattern attacks. The query processing using Arx as efficient as cleartext version due to using an index.

The use of QB with Arx makes it secure against output-size, frequency-count, and workload-skew attacks. Of course, QB takes more time as compared to Arx, since the time of $|\mathit{SB}|$ searches cannot be absorbed in a single index scan unless all $|\mathit{SB}|$ values lie in a single node of the index. In the worst case, we traverse the index at most $|\mathit{SB}|$ times, unlike Arx, which traverses the index only once for a single selection query. 
It, however, significantly enhances the security of Arx by preventing output size, frequency count, and workload-skew attacks. However, QB does not protect access-patterns being revealed which could be prevented using ORAM. Determining whether coupling ORAM with Arx mixed with QB or using a more secure cryptographic solution, \textit{e}.\textit{g}., secret-sharing, which uses a linear scan 
to prevent access-patterns, with QB, more efficient (while QB with both the solutions strengthen the underlying cryptographic technique) is an open question.

{

\section{Experimental Evaluation}
\label{sec:Experimental Validation}
This section presents experimental evaluations of \textsc{Panda}.
As we mentioned that \textsc{Panda} does not need any specific cryptographic technique for encrypted the data, we used \textsc{Panda} with two cryptographic techniques/systems, such as SGX-based Opaque~\cite{opaque} and multi-party computations (MPC)-based Jana~\cite{jana}. We installed Jana on a machine of 3.5GHz, Intel Xeon 8-core processor, 64GB RAM, 3TB disk and Opaque on a machine 3.5GHz, Intel i7 4-core processor, 32GB RAM with SGX support, 350GB disk.\footnote{Please note that we selected two different machines, because our intention is not to compare Jana and Opaque.} We used Order and LineItem tables of TPC-H benchmark. The DB owner stores sensitive and non-sensitive bins, whose size was propositional to the domain size of the searchable attributes and independent of the database size. For example, for LineItem table, metadata for attributes OrderKey, PartKey, and SuppKey were $\approx$3MB, $\approx$1MB and $\approx$0.1MB, respectively, while the size of LineItem table having 12M rows was $\approx$1.5GB.

In the following, we show our experimental results: (\textit{i}) We evaluate \textsc{Panda}'s QB mixed with Jana on 1M LineItem for selection and range queries.\footnote{Data insertion time was significant in Jana, thus, we used only 1M rows. Further, the current Jana version does not support joins in MPC.} (\textit{ii}) We evaluate \textsc{Panda}'s QB mixed with Opaque on 6M and 12M LineItem for selection, range, and join queries.

\begin{table}[!h]
  \centering
    \begin{tabular}{|l|l|l|l|l|l|l|}
    \hline
    \textbf{Technique} & \textbf{20\%} & \textbf{40\%} & \textbf{60\%} & \textbf{80\%}\\ \hline\hline

    MPC-based Jana~\cite{jana} (1M) & 168 & 318 & 481 & 661  \\ \hline

    SGX-based Opaque~\cite{opaque} (6M) & 26 & 42 & 59 & 78 \\ \hline
    \end{tabular}
    \caption{Exp 1: Time (in seconds) for executing a selection query, when mixing QB with Opaque and Jana at different levels of sensitivity.}
    \label{tab:opaque mpc}
\end{table}

\smallskip
\noindent\textbf{Exp 1: Selection query execution.} Table~\ref{tab:opaque mpc} shows the time taken when using QB with Opaque and Jana at different levels of sensitivity. Without using QB for answering a selection query, Opaque~\cite{opaque} took 89 seconds on 6M rows of LineItem table and Jana~\cite{jana} took 797 seconds on 1M rows of LineItem table.\footnote{Note that in the conference version of this paper, the execution time of Jana was higher than the time taken in Table~\ref{tab:opaque mpc}. The reason is that in this paper, we used a newer version of Jana.} Note that the time to execute the same query on cleartext data was only 0.0002 seconds. QB improves not only the performance of Opaque and Jana, but also makes them to work securely on partitioned data and resilient to output-size attack. 

\smallskip	
\noindent\textbf{Exp 2: Range query execution.} We executed a range query on OrderKey column of LineItem table. Opaque and Jana scan the entire data for answering any query. Thus, range query execution time was not impacted by the size of a range. We selected a range of length 400. Table~\ref{tab:range_query_exe_time} shows the time taken by a range query on different sensitivity datasets using QB mixed with Jana and Opaque. Note that the time taken by the same range query without using QB on 1M rows was 841s using Jana, and was 124s on 6M rows and 288s on 12M rows using Opaque. It shows that though fetching more rows according to QB, it does not incur overheads in terms of the computation time.

\begin{table}[!h]
\centering
\begin{tabular}{|l|l|l|l|l|l|l|}\hline
\textbf{Technique} & \textbf{20\%} & \textbf{40\%} & \textbf{60\%} & \textbf{80\%}\\ \hline\hline

MPC-based Jana~\cite{jana} (1M) & 170 & 319 & 483 & 688  \\ \hline

SGX-based Opaque~\cite{opaque} (6M)& 25 & 50 &  74 & 98  \\ \hline

SGX-based Opaque~\cite{opaque} (12M)& 49 & 98 & 146 & 216 \\ \hline

\end{tabular}
\caption{Exp 2: Time (in seconds) for executing a range query, when mixing QB with Opaque and Jana at different levels of sensitivity.}
\label{tab:range_query_exe_time}
\end{table}

\smallskip	
\noindent\textbf{Exp 3: Join query execution.} We executed a join query mixed with a selection predicate covering 400 OrderKey of Order and LineItem tables. We used Order table of 1.5M rows with 6M rows of LineItem table as smaller datasets, and Order table of 3M rows with 12M rows of LineItem table as larger datasets. Table~\ref{tab:join_query_exe_time} shows the time taken by a join query on different sensitivity datasets using QB mixed with Opaque. Note that the time taken by the same join query without using QB was 154s on the smaller dataset and 364s on the larger dataset using Opaque. We also tried to execute a join query without selection; however, Opaque does not support neither amount of rows.

\begin{table}[!h]
\centering
\begin{tabular}{|l|l|l|l|l|l|l|}\hline
\textbf{Technique} & \textbf{20\%} & \textbf{40\%} & \textbf{60\%} & \textbf{80\%}\\ \hline\hline

SGX-based Opaque~\cite{opaque} (6M) & 51 & 77 & 102 & 129 \\ \hline

SGX-based Opaque~\cite{opaque} (12M) & 102 & 155& 207 & 284 \\ \hline

\end{tabular}
\caption{Exp 3: Time (in seconds) for executing a join query, when mixing QB with Opaque at different levels of sensitivity.}
\label{tab:join_query_exe_time}
\end{table}

\smallskip	
\noindent\textbf{Exp 4: Impact of communication cost.}
Since QB fetches more data than the desired data, it may affect the overall performance. We fetched the maximum number of rows in the case of range queries. Particularly, in the case of 80\% of 12M LineItem table, we fetched $\approx$70,000 rows, whose size was $\approx$14MB.  When using slow (100MB/s), medium (500MB/s), and fast (1GB/s) speed of data transmission, the data transmission time was negligible.

%
%
%
%

\smallskip
\noindent\textbf{Exp 5: Impact of bin size.} Table~\ref{tab:Impact of bin-size} plots an average time for a selection query using QB with a different bin size, which is in turn governed by the values of $|\mathit{SB}|$ and $|\mathit{NSB}|$, respectively. We plot the effect of $||\mathit{SB}|-|\mathit{NSB}||$ on retrieval time and find that the minimum time is achieved when $|\mathit{SB}|\approx|\mathit{NSB}|$. 

\begin{table}[!h]
\centering
\begin{tabular}{|l|l|l|l|l|l|l|}\hline
\textbf{$||\mathit{SB}|-|\mathit{NSB}||$} & \textbf{400} & \textbf{500} & \textbf{750} & \textbf{1,000}  \\\hline\hline

Time to execute a selection query on 20\% dataset using Opaque & 25 & 28 & 31 & 34 \\ \hline
\end{tabular}
\caption{Exp 5: Impact of bin-size.}
\label{tab:Impact of bin-size}
\end{table}

\smallskip
\noindent\textbf{Exp 6: Impact of insert operation.} In the experiment, insertions were processed (as per the method, given in \S\ref{subsec:Insert Operation and Re-binning}) in batches of 10,000 and after each batch, selection queries were executed to determine the overhead due to insertion. Finally, after 7 batches of insertion, Algorithm~\ref{alg:bin_creation} was re-executed to recreate bins. Table~\ref{tab:Impact of insert} confirms that the query cost increases but only marginally in the presence of insertion and (as shown by the last column) reduces by re-binning. In Table~\ref{tab:Impact of insert}, the first row shows the size of increasing data after each 10,000 rows' insertion and the second row shows the time in executing a selection on the dataset.

\begin{table}[!h]
\centering
\begin{tabular}{|p{5cm}|l|l|l|l|l|l|l|l|l|}\hline
\textbf{\#inserted rows}  & \textbf{10,000} & \textbf{20,000} & \textbf{30,000} & \textbf{40,000} & \textbf{50,000} & \textbf{60,000} & \textbf{70,000} & \textbf{Re-bin}\\\hline\hline

Time to execute a selection query on 20\% dataset using Opaque & 27 & 28 & 29 & 30 & 31 & 32 & 34 & 30 \\ \hline
\end{tabular}
\caption{Exp 6: Impact of insert.}
\label{tab:Impact of insert}
\end{table}

\smallskip
\noindent\textbf{Exp 7. Number of fake tuples.} Table~\ref{tab:Number of fake tuple inserted due to QB} summarizes the number of fake tuples added for TPC-H LineItem data at different levels of sensitivity. The reason of decreasing fake tuples when increasing sensitivity is that more real tuples take place of the fake tuples. In general, the addition of fake tuples will adversely affect QB, especially, if data is skewed. However, as shown in Tables~\ref{tab:opaque mpc},~\ref{tab:range_query_exe_time}, and~\ref{tab:join_query_exe_time}, QB remains significantly better compared to fully cryptographic approaches at all levels of sensitivity despite fake tuples being added.

\begin{table}[h]
  \centering
  \begin{tabular}{|l|l|l|l|l|l|}
  \hline
  \textbf{LineItem entire size} & \textbf{1\%} & \textbf{5\%} & \textbf{20\%} & \textbf{40\%} & \textbf{60\%}  \\\hline\hline
  6M & 34244 & 34048  & 29568 & 24024 & 22736 \\  \hline
  \end{tabular}
\caption{Exp 7: Number of fake tuple inserted due to QB.}
\label{tab:Number of fake tuple inserted due to QB}
\end{table}

}

\section{Conclusion}
\label{sec:Conclusion}
This paper proposes a prototype, \textsc{Panda}, and its query processing technique, query binning (QB), that serves as a meta approach on top of existing cryptographic techniques to support secure selection, range, and join queries, when a relation is partitioned into cryptographically secure sensitive and cleartext non-sensitive sub-relations. Further, we develop a new notion of \emph{partitioned data security} that restricts exposing sensitive information due to the joint processing of the sensitive and non-sensitive relations. 
Besides improving efficiency, while supporting partitioned security, interestingly, \textsc{Panda} enhances the security of the underlying cryptographic technique by preventing size, frequency-count, and workload-skew attacks. As a result, combining \textsc{Panda}'s QB with efficient but non-secure cloud-side indexable cryptographic approaches can result in an efficient and significantly more secure search. Furthermore, existing indexable/non-indexable cryptographic techniques that prevent access-patterns can also benefit from the added security that \textsc{Panda} offers.  {In future, one may extend the proposed technique for answering queries in different situations, such as the cases of different relations that are encrypted using different cryptographic techniques and the case of a single relation that is vertically partitioned into multiple relations that are encrypted using different cryptographic techniques.}


\appendix

\section{Pseudocode}
\label{app_sec:Pseudocode}

\DontPrintSemicolon
\LinesNotNumbered
\begin{algorithm}[!t]
{
\textbf{Inputs:} $S$ and $\mathit{NS}$

\textbf{Outputs:} $\mathit{SB}[\ast,\ast,\ast]$: Sensitive bins for each level of the tree.
$\mathit{NSB}[\ast,\ast,\ast]$: non-sensitive bins for each level of the tree. In both $\mathit{SB}[\ast,\ast,\ast]$ and $\mathit{NSB}[\ast,\ast,\ast]$, the first index refers to the level of tree, and  the second index refers to the bin identity.

\textbf{Variable Initialization:} An array $\mathit{unique\_ns}[]\leftarrow \mathit{allocate\_distinct\_nonsensitive\_value(\mathit{NS})}$,

$\mathit{array\_tree\_node}[]$: An array to store all the nodes at a particular level of the binary tree.

\textbf{Tree Node Initialization:}
 $\mathrm{TreeNode}(\mathit{covered\_val[]}, \mathrm{TreeNode} \: \mathit{left\_child}, \mathrm{TreeNode} \: \mathit{right\_child})$ {\scriptsize \tcp{$\mathrm{TreeNode}$ is an object class representing a node of a binary tree, where $\mathit{covered\_val[]}$ refers to an array of values covered by the tree node, $\mathit{left\_child}$ refers to the pointer to the child tree node on the left, and $\mathit{right\_child}$ refers to the pointer to the child tree node on the right side of the tree node.}}

\nl{\bf Function $\mathit{create\_binary\_tree(\mathit{unique\_ns}[])}$} \nllabel{ln:function_create_bin_binary_tree}
\Begin{

\nl \If{$|\mathit{unique\_ns}[]| == 1$}{
{\scriptsize \tcp{This `if' statement is the terminal condition of the recursive function execution and the condition for edge case where the number of elements in $\mathit{unique\_ns}[]$ is 1 }}

\nl $\mathrm{TreeNode} \: \mathit{tn} \leftarrow \mathit{allocate}(\mathrm{TreeNode}())$

\nl $\mathit{tn.covered\_val}[0] \leftarrow {\mathit{unique\_ns}}[0]$

\nl $\Return \: tn$}

\nl $\mathit{length}\leftarrow |\mathit{unique\_ns}[]|$ {\scriptsize \tcp{Allocating the count of unique non-sensitive values to a temporary variable $\mathit{length}$.}}

\nl $\mathrm{TreeNode} \: \mathit{root} \leftarrow  \mathit{allocate}(\mathrm{TreeNode}())${{\scriptsize \tcp{Allocate an empty tree node.}}}

\nl \lFor{$i \in (1,\mathit{length})$}{$\mathit{root.covered\_val}[i] \leftarrow {\mathit{unique\_ns}}[i]$} {\scriptsize \tcp{For a tree node, this statement stores all the values covered by the sub-tree rooted at this node}}

\nl $\mathit{root.left\_child} \leftarrow \mathit{create\_binary\_tree}(\mathit{unique\_ns}[1,\mathit{length}/2])$ {\scriptsize \tcp{Recursively calling $\mathit{create\_binary\_tree}()$ function with left half of $\mathit{unique\_ns}[]$ array.}}

\nl $\mathit{root.right\_child} \leftarrow \mathit{create\_binary\_tree}(\mathit{unique\_ns}[\mathit{length}/2+1, \mathit{length}])$ {{\scriptsize \tcp{Recursively calling  $\mathit{create\_binary\_tree}()$ function with right half of $\mathit{unique\_ns}[]$ array.}}}

\nl $\Return \mathit{root}$
}

\nl{\bf Function $\mathit{create\_sensitive\_bin}(\mathrm{TreeNode} \: \mathit{root}, S, \mathit{NS})$ {\scriptsize \tcp{Allocating sensitive and non-sensitive values to bins.}}}
\Begin{

\nl $max\_level \leftarrow \mathit{height}(root)$ {\scriptsize {\tcp{Finding the height of the tree created by function $\mathit{create\_binary\_tree}()$ }}}  \nllabel{ln:find_tree_height_range}

\nl \For{$i \in (1, max\_level -2)$}{
{\scriptsize {\tcp{This `for loop' is used to create bins for each level of the tree, except the root node and child nodes of the root node. This `for loop' starts from the leaf level of the tree.}}}

\nl $\mathit{array\_tree\_node}[] \leftarrow \mathit{retrieveNodes(root,i)}$; {\scriptsize \tcp{
$\mathit{retrieveNodes(root,i)}$ is a function that takes the `root' node of the tree and the level $i$ of  the tree, and then, returns all nodes at the $i^{\mathit{th}}$ level of the tree.}}

\nl $x,y \leftarrow \mathit{approx\_sq\_factors}(|\mathit{array\_tree\_node}[]|)$; $x\geq y$

\nl $|\mathit{NSB}| \leftarrow x$, $\mathit{NSB} \leftarrow \lceil |\mathit{array\_tree\_node}[]|/x\rceil$, $\mathit{SB} \leftarrow x$, $|\mathit{SB}| \leftarrow y$
{\scriptsize \tcp{Sensitive and non-sensitive bin creation for the tree nodes of the $i^{\mathit{th}}$-level of the tree}}

\nl \lFor{$(j,k) \in (1,\mathit{NSB}), (1,|\mathit{NSB}|)$}{$\mathit{NSB}[i][j][k]\leftarrow \mathit{array\_tree\_node}[(j-1)*|NSB|+k]$\nllabel{ln:non-sesnitive_allocate_range}}
{\scriptsize \tcp{Allocating non-sensitive values to bins created for nodes of the $i^{\mathit{th}}$-level of the tree.}}

\nl \lFor{$(j,k)\in (1,\mathit{NSB}),(1,|\mathit{NSB}|)$}{$\mathit{SB}[i][k][j]\leftarrow \mathit{allocateS(\mathit{NSB}[i][j][k])}$ \nllabel{ln:assign_value_to_NS_bucket_range}}
{\scriptsize \tcp{Allocating sensitive values to bins created for nodes of the $i^{\mathit{th}}$-level of the tree.}}

\nl \lFor{$j\in (1,\mathit{SB})$}{$\mathit{SB}[j][\ast]\leftarrow$ fill the bin if empty with the size limit to $y$  \nllabel{ln:assign_remaining_S_range}}
{\scriptsize \tcp{Filling the sensitive bins for the $i^{\mathit{th}}$ level with the sensitive values that are non-associated with any non-sensitive value.}}

\nl \Return $\mathit{SB}[i][\ast][\ast]$ and $\mathit{NSB}[i][\ast][\ast]$
}
}
\caption{Bin-creation algorithm for range queries.}
\label{alg:bin_creation_range}
}
\end{algorithm}
\setlength{\textfloatsep}{0pt}

\DontPrintSemicolon
\LinesNotNumbered
\begin{algorithm}[!t]
{
\textbf{Inputs:}
[$\alpha, \beta$]: The range of query values, where $\beta>\alpha$ and $\alpha,\beta \in \mathit{NS}$.
$\mathit{Tree}$: A binary tree created by Algorithm~\ref{alg:bin_creation_range}.

\textbf{Outputs:} $\mathit{SB}$: A sensitive bin; $\mathit{NSB}$: A non-sensitive bin.

\textbf{Variable Initialization:} $\mathit{iterations}\leftarrow \mathit{height}(\mathit{Tree})$. $\mathit{found} \leftarrow \mathit{false}$.

\nl{\bf Function $\mathit{retrieve\_bins\_binary\_tree(\alpha, \beta)}$} \nllabel{ln:function_create_bin_tree}
\Begin{

\nl \For{$i \in (2 , \mathit{iterations})$ \scriptsize{\tcp{Traversing the tree from the parent nodes of the leaf node to the root node in a bottom-up fashion.}}}{

\nl $\mathit{array\_tree\_node}[] \leftarrow \mathit{retrieveNodes(Tree.root,i)}${\scriptsize \tcp{
$\mathit{retrieveNodes(root,i)}$ is a function that takes the `root' node of the tree and the level $i$ of  the tree, and then, returns all nodes at the $i^{\mathit{th}}$ level of the tree.}}

\nl \For{$j \in (1,|\mathit{array\_tree\_node}[]|$}{
{\scriptsize \tcp{This `for loop' iterates over the nodes of $\mathit{array\_tree\_node[]}$ array and determines if any node covers the range of [$\alpha,\beta$] completely.}}

\nl $\mathit{len\_covered} \leftarrow |\mathit{array\_tree\_node[j].covered\_val[]}|$

\nl  \lIf{($\mathit{array\_tree\_node[j].covered\_val[0]} \leq \alpha \wedge  \mathit{array\_tree\_node[j].covered\_val[len\_covered] } \geq \beta$) {\scriptsize \tcp{If the range  [$\alpha$, $\beta$] is covered completely by $\mathit{covered\_val}[]$ array of the $j^{\mathit{th}}$ node at  level $i$ of the tree, which is stored in $\mathit{array\_tree\_node[j]}$.}}}
 {

 \nl $\mathit{NSB.append(array\_tree\_node[j].covered\_val[\ast])}$ {\scriptsize \tcp{All the values covered by the node $\mathit{array\_tree\_node[j]}$ are added to the non-sensitive bin $\mathit{NSB}.$}}

\nl  $\mathit{found} \leftarrow true$; $\mathit{level} \leftarrow i$; $\mathit{bin\_id} \leftarrow j$;
 $Break$
}
}
\nl \lIf{$\mathit{found} = true$}{$Break$}
}
\nl \Return $\mathit{SB[level,bin\_id,\ast]}$ and $\mathit{NSB}$
}
\caption{Bin retrieval algorithm for range queries (best-match method) for non-sensitive values.}
\label{alg:bin_retrieve_range_best_match}
}
\end{algorithm}
\setlength{\textfloatsep}{0pt}

\DontPrintSemicolon
\LinesNotNumbered
\begin{algorithm}[!t]
{
\textbf{Inputs:} $\mathit{NS}$: non-sensitive data, $S$: sensitive data.

\textbf{Outputs:} $\mathit{SB}$: sensitive bins; $\mathit{NSB}$: non-sensitive bins

\textbf{Variable:}
$\mathit{frequent\_non\_sensitive\{\}}$: A set holding frequent non-sensitive keywords

\nl{\bf Function $\mathit{create\_bins\_under\_workload(S,NS)}$}\nllabel{ln:function_create_bucket_workload}
\Begin{
    \nl $x, y \leftarrow \mathit{approx\_sq\_factors(|NS|)}: x \geq y$ {\scriptsize \tcp{Finding approximately square factors of the count of unique non-sensitive values.}} \nllabel{ln:largest_divisors_workload}

    \nl $|\mathit{NSB}| \leftarrow x$, $\mathit{NSB} \leftarrow \lceil |\mathit{NS}|/x\rceil$, $\mathit{SB} \leftarrow x$, $|\mathit{SB}| \leftarrow y$ {\scriptsize \tcp{Creating sensitive and non-sensitive bins.}} \nllabel{ln:number_of_buckets_workload}

    \nl $\mathit{frequent\_non\_sensitive\{\}} \leftarrow \mathit{find\_frequent(NS)}$
    {\scriptsize \tcp{Find and allocate frequent non-sensitive values to the set $\mathit{frequent\_non\_sensitive\{\}}$.}}
    \nllabel{ln:find_frequent_ns}

    \nl  \For{$(i,j) \in (1,\mathit{NSB}),(1,|\mathit{frequent\_non\_sensitive}\{\}|)$}{

    $\mathit{NSB}[i][j]\leftarrow \mathit{frequent\_non\_sensitive[(i-1)*|\mathit{NSB}|+j]}$
    {\scriptsize \tcp{Allocating frequent non-sensitive query keywords to non-sensitive bins.}}
    \nllabel{ln:non_sensitive_allocate_workload}
      }

    \nl $\mathit{remaining\_NS}\{\} \leftarrow \mathit{NS} \setminus
    \{v \mid v \in \mathit{frequent\_non\_sensitive\{\}\}} ${\scriptsize \tcp{The set $\mathit{remaining\_NS}\{\}$ contains all the infrequent non-sensitive values.}}\nllabel{ln:find_remaining_non-sensitive_workload}

    \nl
    \lFor{$i \in (1,\mathit{NSB})$}{$\mathit{NSB[i][\ast]}\leftarrow$ fill the bin if empty with the size limit to $x$ by taking values from $\mathit{remaining\_NS}\{\}$ until $\mathit{remaining\_NS}\{\}=\emptyset$}

\nl \lFor{$(i,j)\in (1,\mathit{NSB}),(1,|\mathit{NSB}|)$}{
$\mathit{SB}[j][i]\leftarrow \mathit{allocateS(\mathit{NSB}[i][j])}$} {\scriptsize \tcp{Assigning sensitive values associated with non-sensitive values.}}
\nllabel{ln:assign_value_to_S_bucket_workload}

\nl $\mathit{remaining\_S}\{\} \leftarrow S\setminus \{v \mid v \in \mathit{SB[\ast,\ast]}\}$ {\scriptsize \tcp{The set $\mathit{remaining\_S}\{\}$ contains all the sensitive values that are not associated with any non-sensitive values.}} \nllabel{ln:find_remaining_sensitive_workload}

\nl \lFor{$i\in (1,\mathit{SB})$}{$\mathit{SB}[i,\ast]\leftarrow$ fill the bin if empty with the size limit to $y$ by taking values from $\mathit{remaining\_S}\{\}$ until $\mathit{remaining\_S}\{\}=\emptyset$ \nllabel{ln:assign_remaining_S_workload}}

\nl \Return $\mathit{SB}$ and $\mathit{NSB}$
}
\caption{Bin-creation algorithm to deal with the workload-skew attack.}
\label{alg:bin_creation_workload}
}
\end{algorithm}
\setlength{\textfloatsep}{0pt}

\end{document}